\documentclass[a4paper,USenglish]{lipics-v2021}

\pdfoutput=1 
\hideLIPIcs  

\title{Complexity of the (Connected) Cluster Vertex Deletion problem on $H$-free graphs
\thanks{Parts of this paper was presented at the 47th International Symposium on Mathematical Foundations of Computer Science (MFCS 2022)~\cite{LeLe22}.}
}

\titlerunning{Cluster Vertex Deletion on $H$-free graphs}

\author{Hoang-Oanh Le}{Independent Researcher, Berlin, Germany}{HoangOanhLe@outlook.com}{}{}
\author{Van Bang Le}{Institut f\"ur Informatik, Universit\"at Rostock, Germany}{van-bang.le@uni-rostock.de}{}{}

\authorrunning{Hoang-Oanh Le and Van Bang Le}

\Copyright{Hoang-Oanh Le and Van Bang Le}

 

\category{} 

\acknowledgements{We are grateful to the reviewers for their careful reading and helpful comments. In particular, we thank one of them for her/his very meticulous reading with many valuable suggestions that significantly improved the quality of the paper.}

\nolinenumbers 

\usepackage{dsfont} 
\usepackage{mathtools}
\allowdisplaybreaks

\usepackage[usenames,dvipsnames,table]{xcolor}

\usepackage{tikz}
\usetikzlibrary{shapes,calc,positioning,arrows}

\usepackage[linesnumbered,boxed,ruled]
       {algorithm2e} 
\SetAlgoInsideSkip{medskip}

\def\NP{\mathsf{NP}}
\def\VC{\textup{\textsc{vertex cover}}}
\def\CVC{\textup{\textsc{connected vertex cover}}}
\def\CVD{\textup{\textsc{cluster-vd}}}
\def\CCVD{\textup{\textsc{connected cluster-vd}}}
\def\3SAT{\textup{\textsc{3sat}}}
\def\NAE3SAT{\textup{\textsc{nae 3sat}}}
\def\2IN3SAT{\textup{\textsc{1-in-3 3sat}}}
\def\pq3SAT{\textup{\textsc{$(p,q)$-3sat}}}
\def\pqNAE3SAT{\textup{\textsc{$(p,q)$-nae 3sat}}}

\def\join{{\footnotesize\textcircled{\textoneoldstyle}}}
\def\cojoin{{\footnotesize\textcircled{\textzerooldstyle}}}

\tikzstyle{knoten}=[circle,inner sep=0pt]
\tikzstyle{vertex}=[draw,circle,inner sep=1.2pt] 
\tikzstyle{square}=[draw,rectangle,inner sep=1.5pt] 
\tikzstyle{vertexS}=[draw,circle,inner sep=1.2pt,fill=black] 
\tikzstyle{squareS}=[draw,rectangle,inner sep=1.5pt,fill=black]

\begin{document}

\maketitle

\begin{abstract}
The well-known Cluster Vertex Deletion problem (\CVD) asks for a given graph $G$ and an integer $k$ whether it is possible to delete a set $S$ of at most $k$ vertices of $G$ such that the resulting graph $G-S$ is a cluster graph (a disjoint union of cliques).  
We give a complete characterization of graphs $H$ for which \CVD\ on $H$-free graphs is polynomially solvable and for which it is $\NP$-complete. Moreover, in the $\NP$-completeness cases, \CVD\ cannot be solved in sub-exponential time in the vertex number of the $H$-free input graphs unless the Exponential-Time Hypothesis fails.
We also consider the connected variant of \CVD, the Connected Cluster Vertex Deletion problem (\CCVD), in which the set $S$ has to induce a connected subgraph of $G$. It turns out that \CCVD\ admits the same complexity dichotomy for $H$-free graphs. 
Our results enlarge a list of rare dichotomy theorems for well-studied problems on $H$-free graphs.
\end{abstract}

\section{Introduction and results}
A very extensively studied version of graph modification problems asks to modify a given graph to a graph that satisfies a certain property $\cal G$ by deleting a minimum number of vertices. 
The case $\cal G$ being \lq edgeless\rq\ is the well-known \VC\ problem, one of the classical $\NP$-hard problems.  
If $\cal G$ is a \lq cluster graph\rq, a graph in which every connected component is a clique, the corresponding problem is another well-known $\NP$-hard problem, the \textsc{cluster vertex deletion} problem (\CVD\ for short). 
In this paper, we revisit the computational complexity of \CVD, formally given below. 

\medskip\noindent
\fbox{
\begin{minipage}{.96\textwidth} 
\CVD\\[.5ex]
\begin{tabular}{l l}
{\em Instance:}& A graph $G=(V,E)$ and an integer $k$.\\
{\em Question:}& Does there exist a vertex set $S\subseteq V$ of size at most $k$ such that $G-S$\\ 
               & is a cluster graph?
\end{tabular}
\end{minipage}
}

\medskip
Being an hereditary property on induced subgraphs, \CVD\ is $\NP$-complete~\cite{LewisY80} and cannot be solved in $2^{o(n+m)}$ time unless the ETH (Exponential-Time Hypothesis) fails~\cite{Komusiewicz18}, where $n$ and $m$ are the vertex and edge number of the input graphs, respectively.
\CVD\ remains $\NP$-complete even when restricted to planar graphs~\cite{Yannakakis78} and to bipartite graphs~\cite{Yannakakis81a}, and to planar bipartite graphs of maximum degree~$3$~\cite{HsiehLLP22}. 
Most recent works on \CVD\ deal with exact, {FPT} and approximation algorithms~\cite{AprileDFH22,BoralCKP16,HuffnerKMN10,Tsur21}. 

It is noticeable that there are only a few known cases where the problem can be solved efficiently: \CVD\ is polynomially solvable on block graphs, split graphs and interval graphs~\cite{0001KOY18}, and on graphs of bounded treewidth~\cite{SauS21}. On the other hand, the complexity status of \CVD\ on many well-studied graph classes is still open, e.g., chordal graphs discussed in~\cite{0001KOY18} and planar bipartite graphs mentioned in~\cite{ChakrabortyCPP21}.

In this paper we initiate studying the computational complexity of \CVD\ on graphs defined by forbidding certain induced subgraphs. We remark that related approaches for other problems are quite common in the literature, e.g., for \VC\ (aka \textsc{independent set})~\cite{
GartlandL20,GrzesikKPP22} and \textsc{coloring}~\cite{GolovachJPS17,KralKTW01}, and that many popular graph classes are defined or characterized by forbidding induced subgraphs, e.g., chordal and bipartite graphs (by infinitely many forbidden subgraphs), and cographs and line graphs (by finitely many forbidden subgraphs).

All graphs considered are undirected, finite and have no multiple edges or self-loops. Let~$H$ be a given graph. A graph $G$ is \emph{$H$-free} if no induced subgraph in $G$ is isomorphic to~$H$. A path with $n$ vertices and $n-1$ edges is denoted by $P_n$.  
The main result of the present paper is the following complexity dichotomy:
\begin{theorem}\label{thm:H-free}
Let $H$ be a fixed graph. \CVD\ is polynomially solvable on $H$-free graphs if $H$ is an induced subgraph of the $4$-vertex path $P_4$, and $\NP$-complete otherwise. 

Furthermore, in case $H$ is not an induced subgraph of $P_4$, no algorithm of runtime $2^{o(n)}$ can solve \CVD\ on $H$-free $n$-vertex graphs, unless the ETH fails.
\end{theorem}

We also consider the connected variant of \CVD, which is as follows.

\medskip\noindent
\fbox{
\begin{minipage}{.96\textwidth} 
\CCVD\\[.5ex]
\begin{tabular}{l l}
{\em Instance:}& A graph $G=(V,E)$ and an integer $k$.\\
{\em Question:}& Does there exist a vertex set $S\subseteq V$ of size at most $k$ such that $G-S$\\ 
               & is a cluster graph and $G[S]$ is connected?
\end{tabular}
\end{minipage}
}

\medskip\noindent
It is known that \CCVD\ is $\NP$-complete and cannot be solved in $2^{o(n+m)}$ time unless the ETH fails~\cite{Komusiewicz18}. 
It turns out that \CCVD\ admits the same complexity dichotomy as for \CVD:
\begin{theorem}\label{thm:H-free-connected}
Let $H$ be a fixed graph. \CCVD\ is polynomially solvable on $H$-free graphs if $H$ is an induced subgraph of the $4$-vertex path $P_4$, and $\NP$-complete otherwise. 

Furthermore, in case $H$ is not an induced subgraph of $P_4$, no algorithm of runtime $2^{o(n)}$ can solve \CCVD\ on $H$-free $n$-vertex graphs, unless the ETH fails.
\end{theorem}

Theorems~\ref{thm:H-free} and~\ref{thm:H-free-connected} enlarge a list of rare dichotomy theorems on $H$-free graphs: Korobitsin~\cite{Korobitsin92} proved that \textsc{dominating set} is solvable in polynomial time on $H$-free graphs if~$H$ is an induced subgraph of $P_4+tP_1$, the union of $P_4$ and $t$ isolated vertices for $t\ge0$, and $\NP$-complete otherwise. Munaro~\cite{Munaro17} proved that the same dichotomy holds for \textsc{connected dominating set} and for \textsc{graph VC$_{\text{con}}$ dimension}. 
Kr\'al, Kratochv{\'{\i}}l, Tuza and Woeginger~\cite{KralKTW01} proved that \textsc{colouring} on $H$-free graphs is solvable in polynomial time if~$H$ is an induced subgraph of $P_4$ or of $P_3+P_1$ and $\NP$-complete otherwise. Kami\'nski~\cite{Kaminski12} proved that \textsc{max-cut} is solvable in polynomial time if $H$ is an induced subgraph of $P_4$ and $\NP$-complete otherwise.

\section{Preliminaries}
For a set $\cal H$ of graphs, ${\cal H}$-free graphs are those in which no induced subgraph is isomorphic to a graph in $\cal H$. We denote by $K_{1,n}$ the tree with $n+1\ge 3$ vertices and $n$ leaves, by $C_n$ the $n$-vertex cycle. The girth $girth(G)$ of a graph~$G$ is the smallest length of a cycle in~$G$; we set $girth(G)=\infty$ if~$G$ is a \emph{forest}, a graph without cycles. Thus, for any fixed integer $g\ge 3$, $gith(G)>g$ if and only if~$G$ is $\{C_3,C_4,\allowbreak\ldots,\allowbreak C_g\}$-free. 

As usual, we denote by $\overline{G}$ the complement of a graph~$G$. 
The union $G+H$ of two vertex-disjoint graphs $G$ and $H$ is the graph with vertex set $V(G)\cup V(H)$ and edge set $E(G)\cup E(H)$; we write $pG$ for the union of~$p$ copies of~$G$. For a subset $S \subseteq V(G)$, let $G[S]$ denote the subgraph of $G$ induced by~$S$; $G-S$ stands for $G[V(G)\setminus S]$.  
By \lq $G$ contains an $H$\rq\ we mean~$G$ contains~$H$ as an induced subgraph. Graphs in which every vertex has degree~$3$ are called \emph{$3$-regular graphs} or \emph{cubic graphs} and graphs with maximum degree~$3$ \emph{subcubic graphs}.

A graph $G$ is a \emph{cluster graph} if each of its connected components is a clique. Observe that~$G$ is a cluster graph if and only if~$G$ is $P_3$-free. If $S\subseteq V(G)$ is a subset of vertices of $G$ such that $G-S$ is $P_3$-free, then~$S$ is called a \emph{cluster vertex deletion set} of~$G$. An \emph{optimal} cluster vertex deletion set is one of minimum size. 

Algorithmic lower bounds in this paper are conditional, based on the Exponential Time Hypothesis (ETH)~\cite{ImpagliazzoP01}. The ETH asserts that no algorithm can solve \3SAT\ in subexponential time $2^{o(n)}$ for $n$-variable \textsc{3-cnf} formulas. As shown by the Sparsification Lemma in~\cite{ImpagliazzoPZ01}, the hard cases of \3SAT\ consist of sparse formulas with $m=O(n)$ clauses. Hence, the ETH implies that \3SAT\ cannot be solved in time $2^{o(n+m)}$. 

Recall that an instance for \NAE3SAT\ is a \textsc{3-cnf} formula $F=C_1\land C_2\land \cdots \land C_m$ over $n$ variables, in which each clause $C_j$ consists of three distinct literals. The problem asks whether there is a truth assignment of the variables such that every clause in $F$ has at least one true and at least one false literal. Such an assignment is called an \emph{nae assignment}, i.e. a not-all-equal assignment. There is a polynomial reduction from \3SAT\ to \NAE3SAT\ (\cite[Theorem 7.3]{Moret98}), which transforms an instance for \3SAT\ with $n$ variables and $m$ clauses to an equivalent instance for \NAE3SAT\ with $2n+24m$ variables and $32m$ clauses. Thus, we obtain: 

\begin{theorem}[\cite{Moret98,ImpagliazzoPZ01}]\label{thm:nae3sat}
\NAE3SAT\ is $\NP$-complete and, assuming ETH, cannot be solved in time $2^{o(n+m)}$ on inputs with $n$ variables and $m$ clauses.
\end{theorem}

We will also need the following restriction of \NAE3SAT. For integers $p, q\ge 2$, let \pq3SAT\ denote the problem of deciding if a \textsc{3-cnf} formula in which each variable occurs at most~$p$ times positively and at most~$q$ times negatively is satisfiable. \pqNAE3SAT\ is defined analogously. A reduction from \3SAT, linear in the number of clauses, due to Tovey~\cite{Tovey84} shows that \textsc{$(2,2)$-3sat} remains $\NP$-complete and, assuming ETH, cannot be solved in time $2^{o(n)}$ time for inputs with $n$ variables. 
Now, the reduction due to Moret~\cite[Theorem 7.3]{Moret98} mentioned above transforms an instance for \textsc{$(2,2)$-3sat} to an equivalent instance for \textsc{$(4,4)$-nae 3sat}, linear in the number of variables and clauses. Hence, we obtain: 

\begin{theorem}[\cite{Tovey84,Moret98,ImpagliazzoPZ01}]\label{thm:(4,4)-nae3sat}
\textup{\textsc{$(4,4)$-nae 3sat}} is $\NP$-complete and, assuming ETH, cannot be solved in time $2^{o(n)}$ on inputs with $n$ variables.
\end{theorem}

\smallskip\noindent
\textbf{Structure of the paper.} 
We first address the polynomial part of Theorems~\ref{thm:H-free} and~\ref{thm:H-free-connected} in the next section. Then we present two new $\NP$-completeness results for \CVD\ and \CCVD\ in Sections~\ref{sec:3P1and2P2} and~\ref{sec:girth}. 
These hardness results allow us to clear the $\NP$-completeness part of Theorems~\ref{thm:H-free} and~\ref{thm:H-free-connected} in Section~\ref{sec:NPc-part}. The last section concludes the paper. 

\section{$H$-free graphs: polynomial cases}\label{sec:P-part}
The polynomial part in Theorems~\ref{thm:H-free} and~\ref{thm:H-free-connected} consists of six cases; see Fig.~\ref{fig:easy} for all graphs~$H$ for which \CVD\ and \CCVD\ are polynomially solvable on $H$-free graphs. 

\begin{figure}[!ht]
\begin{center}

\begin{tikzpicture}[scale=.42] 
\node[vertex] (x1) at (1,3)  {};

\node at (1,1.5) {$P_1$};
\end{tikzpicture} 
\qquad\,
\begin{tikzpicture}[scale=.42] 
\node[vertex] (x1) at (1,3)  {};
\node[vertex] (x2) at (2.5,3)  {}; 

\node at (1.75,1.5) {$2P_1$};
\end{tikzpicture} 
\qquad\,
\begin{tikzpicture}[scale=.42] 
\node[vertex] (x1) at (1,3)  {};
\node[vertex] (x2) at (2.5,3)  {};

\draw (x1)--(x2); 

\node at (1.75,1.5) {$P_2$};
\end{tikzpicture} 
\qquad\,
\begin{tikzpicture}[scale=.42] 
\node[vertex] (x1) at (1,3)  {};
\node[vertex] (x2) at (2.5,3)  {};
\node[vertex] (x3)  at (4,3) {}; 

\draw (x1)--(x2); 

\node at (2.5,1.5) {$P_2+P_1$};
\end{tikzpicture} 
\qquad\,
\begin{tikzpicture}[scale=.42] 
\node[vertex] (x1) at (1,3)  {};
\node[vertex] (x2) at (2.5,3)  {};
\node[vertex] (x3)  at (4,3) {}; 

\draw (x1)--(x2)--(x3); 

\node at (2.5,1.5) {$P_3$};
\end{tikzpicture} 
\qquad\,
\begin{tikzpicture}[scale=.42] 
\node[vertex] (x1) at (1,3)  {};
\node[vertex] (x2) at (2.5,3)  {};
\node[vertex] (x3)  at (4,3) {}; 
\node[vertex] (x4)  at (5.5,3) {};

\draw (x1)--(x2)--(x3)--(x4); 

\node at (3.25,1.5) {$P_4$};
\end{tikzpicture} 
\caption{The graphs $H$ for which \CVD\ and \CCVD\ are polynomially solvable on $H$-free graphs.}\label{fig:easy}
\end{center}
\end{figure}
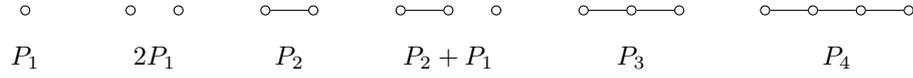

Observe that $H$-freeness is hereditary, meaning if $H'$ is an induced subgraph of $H$ then $H'$-free graphs are $H$-free graphs. Thus, it suffices to prove the polynomial part only for the case where $H$ is the $4$-vertex path $P_4$. 

The proof will follow from the concept of clique-width of graphs in connection with the so-called monadic second-order logic, $MSOL_1$ for short, an extension of first-order logic with quantification over vertex set variables. Briefly, the clique-width of a graph $G$, introduced in~\cite{CourcelleER93}, is the minimum number of labels needed to construct $G$ by:
\begin{itemize}
\item creating a new vertex with label $i$,
\item taking a disjoint union of two labeled graphs,
\item joining every vertex with label $i$ to every vertex with label $j\not= i$, and
\item renaming label $i$ to label $j$. 
\end{itemize}
Such a construction with $k$ labels defines an algebraic $k$-expression. 
A well-known meta-theorem by Courcelle, Makowsky and Rotics~\cite{CourcelleMR00} states that any graph property expressible in $MSOL_1$ is decidable in linear time for graphs with bounded clique-width, provided a $k$-expression of the graphs is given. 
It is well known that $P_4$-free graphs, also known as cographs, have clique-width at most~$2$ and a corresponding $2$-expression can be constructed in linear time (see, e.g.,~\cite{CourcelleMR00}). Hence, any $MSOL_1$ graph property is decidable in linear time when restricted to $P_4$-free graphs. 

Now, being a cluster vertex deletion set is a $MSOL_1$ property:
\begin{align*}
& \forall u, v, w \big(\neg S(u)\land \neg S(v)\land \neg S(w) \land E(u,v)\land E(v,w)\land (u\not= w) \to E(u,w)\big), &
\end{align*}
where $S(x)$ means $x\in S$ and $E(x,y)$ means $xy\in E(G)$. (The sentence says that the graph $G-S$ is $P_3$-free.) 

Also, the fact that the vertex set $S$ in a graph $G$ induces a connected subgraph of $G$ can be written as a $MSOL_1$ sentence:
\begin{align*}
& \forall T\subseteq S \Big((S\not=\emptyset \land S\setminus T\not=\emptyset) \to \big(\exists u\in S\setminus T,\, \exists v\in T:\, E(u,v)\big)\Big).&
\end{align*}
(The sentence says that, for any bipartition of $S$ into two non-empty sets, there is an edge joining two vertices in different parts of the bipartition.) 

Thus, \CVD\ and \CCVD\ can be solved in linear time on $P_4$-free graphs.  
Indeed, we have a stronger fact. The weighted optimization version of \CVD\ and \CCVD,  \textsc{minimum cluster-vd} and \textsc{minimum connected cluster-vd}, are $LinEMSOL_{\tau_{1,p}}$ problems ($LinEMSOL_{\tau_{1,p}}$ is an extension of $MSOL_1$ which allows one to search for optimal sets of vertices with respect to some linear objective function). 
We refer to the paper~\cite{CourcelleMR00} for details, in which it is shown that every $LinEMSOL_{\tau_{1,p}}$ problem on $P_4$-free graphs can be solved in linear time~\cite[Theorem~4]{CourcelleMR00}. To sum up, we have:

\begin{proposition}\label{pro:P4}
\CVD\ and \CCVD\ can be solved in linear time on $P_4$-free graphs, even in the weighted optimization version.
\end{proposition}

Another approach for obtaining the above results is to use the so-called cotree of cographs. Using the cotree of a cograph $G$, we are able to compute an optimal (connected) cluster vertex deletion set of $G$ in linear time in a direct and simple way. The details are given in the appendices~\ref{app:P4free} and~\ref{app:P4free-connected}.

\section{Cluster-VD and Connected Cluster-VD on dense graphs}\label{sec:3P1and2P2}
In this section, we give a polynomial reduction from \VC\ to \CVD, showing that \CVD\ remains $\NP$-complete when restricted to $\{3P_1, 2P_2\}$-free $n$-vertex graphs with minimum degree at least $n-4$.  

Recall that the \VC\ problem asks, for a given graph $G$ and an integer $k$, if one can delete a vertex set $S$ of size at most $k$ such that $G-S$ is edgeless. 
It is well known that \VC\ is $\NP$-complete and, assuming ETH, cannot be solved in $2^{o(n+m)}$ time on $n$-vertex $m$-edge graphs. This fact and a result in~\cite{JohnsonS99} imply that, assuming ETH, \VC\ cannot be solved in $2^{o(n)}$ time on subcubic $n$-vertex graphs. 
There is a polynomial-time reduction from \VC\ in cubic graphs to \VC\ in subcubic planar graphs with arbitrarily large girth, which transforms an instance $(G,k)$ of the first version to an equivalent instance $(G',k')$ for the second version, where the vertex number of $G'$ is linear in the vertex number of $G$ (see, e.g.,~\cite{Murphy92} or~\cite{Komusiewicz18}). Thus, we obtain:

\begin{theorem}[\cite{JohnsonS99,Murphy92,Komusiewicz18}]\label{thm:vc-3}
Let $g\ge 3$ be a fixed integer. \VC\ is $\NP$-complete even when restricted to subcubic graphs of girth $>g$ and, assuming ETH, \VC\ cannot be solved in  $2^{o(n)}$ time in this restricted graph class.
\end{theorem}

We now describe the announced reduction. Let $g\ge 3$ be an integer and let $(G,k)$ be an instance for \VC, where $G$ is a $n$-vertex subcubic graph with girth $>g$. We may assume that 

\begin{itemize}
\item $G$ is not perfect. 
This is because \VC\ is polynomially solvable on perfect graphs (see~\cite{GLS1988}); notice that $G$ is perfect if and only if $\overline{G}$ is perfect and perfect graphs can be recognized in polynomial time~\cite{ChudnovskyCLSV05},  and

\item $k\le |V(G)|/2$. This fact 
can be easily seen as follows: given $G$ with $n$ vertices and an integer $k$, let $G'$ be obtained from $G$ by adding $p=\max\{0,2k-n\}$ isolated vertices. Then $k=|V(G')|/2$ and $(G,k)\in\VC$ if and only if $(G',k)\in\VC$. Notice that like $G$, $G'$ is subcubic, not perfect and has girth~$>g$, too. 
\end{itemize}

From $(G,k)$ we construct an equivalent instance $(G',k')$ for \CVD\ as follows: $G'$ is obtained from two disjoint copies of $\overline{G}$, $G_1$ and $G_2$, by adding all possible edges between $V(G_1)$ and $V(G_2)$. Set $k'=2k$. 

We argue that $(G,k)\in\VC$ if and only $(G',k')\in \CVD$. First, let $S\subset V(G)$ be a vertex cover, that is $G-S$ is edgeless, with $|S|\le k$. Let $S_1$ and $S_2$ be the copy of $S$ in~$G_1$ and~$G_2$, respectively. 
Then, for each $i\in\{1,2\}$, $G_i-S_i$ is a clique in $G_i=\overline{G}$, and with $S'=S_1\cup S_2$, $G'-S'$ is a clique in~$G'$ with $|S'|=2|S|\le 2k=k'$. 

Conversely, let $S'\subseteq V(G')$ be a cluster vertex deletion set of $G'$ with $|S'|\le k'=2k$. 
Observe that, for each $i\in\{1,2\}$, $S'\cap V(G_i)$ is a proper nonempty subset of $V(G_i)$: if for some $i$, $S'\cap V(G_i)=\emptyset$ then $G_i$ (hence $G$) would be perfect because in this case $G_i$ would be a cluster, and if $V(G_i)\subset S'$ then $2k\ge |S'|>|V(G_i)|=|V(G)|$, contradicting $k\le|V(G)|/2$. 
It follows from the above that $G'-S'$ is a single clique, implying for each $i\in\{1,2\}$, $G_i-S_i$ is a clique in $G_i$ where $S_i=S'\cap V(G_i)$. Since $|S'|\le 2k$, $|S_1|\le k$ or $|S_2|\le k$. Let $|S_1|\le k$, say, and let $S\subseteq V(G)$ be the set of the corresponding vertices in $G$. Then $G-S$ is edgeless with $|S|\le k$. 

We have seen that $G$ has a vertex cover of size at most $k$ if and only if $G'$ has a cluster vertex deletion set of size at most $k'$, as claimed.

Note that $G'$ has $2n$ vertices and minimum degree at least~$2n-4$ (as $G$ has $n$ vertices and maximum degree at most~$3$). 
Now, observe that, for any \emph{connected} graph $X$, if $G$ is $X$-free then $G'$ is $\overline{X}$-free.  
Since $G$ is $\{C_3,C_4,\ldots,C_g\}$-free, we obtain with Theorem~\ref{thm:vc-3}:
\begin{theorem}\label{thm:3P1and2P2}
For any fixed $g\ge 3$, \CVD\ is $\NP$-complete on $\{\overline{C_3}, \overline{C_4}, \allowbreak\ldots, \allowbreak\overline{C_g}\}$-free $n$-vertex graphs with minimum degree at least~$n-4$ and, assuming ETH, cannot be solved in $2^{o(n)}$ time. 

In particular, \CVD\ is $\NP$-complete on $\{3P_1, 2P_2\}$-free graphs and, assuming ETH, cannot be solved in $2^{o(n)}$ time. 
\end{theorem}

We observe that the proof of Theorem~\ref{thm:3P1and2P2} remains true for connected cluster vertex deletion sets: $G$ has a vertex cover of size at most $k\le |V(G)|/2$ if and only if $G'$ has a \emph{connected} cluster vertex deletion set of size at most $k'=2k$. Thus, Theorem~\ref{thm:3P1and2P2} also holds for \CCVD:

\begin{theorem}\label{thm:3P1and2P2-ccvd}
For any fixed $g\ge 3$, \CCVD\ is $\NP$-complete on $\{\overline{C_3}, \overline{C_4}, \allowbreak\ldots, \allowbreak\overline{C_g}\}$-free $n$-vertex graphs with minimum degree at least~$n-4$ and, assuming ETH, cannot be solved in $2^{o(n)}$ time. 

In particular, \CCVD\ is $\NP$-complete on $\{3P_1, 2P_2\}$-free graphs and, assuming ETH, cannot be solved in $2^{o(n)}$ time. 
\end{theorem}

\section{Cluster-VD and Connected Cluster-VD on sparse graphs}\label{sec:girth}
In~\cite[Lemma 1]{Yannakakis81a}, Yannakakis gave a polynomial-time reduction from \NAE3SAT\ to \CVD, which transforms an instance for \NAE3SAT\ with $n$ variables and $m$ clauses, into an equivalent instance $(G,k)$ for \CVD, where $G$ is a bipartite graph with $6n+12m$ vertices. 
Thus, by Theorem~\ref{thm:nae3sat}, 
\CVD\ is $\NP$-complete even when restricted to bipartite graphs and, assuming ETH, \CVD\ cannot be solved in $2^{o(n)}$ time on bipartite graphs with $n$ vertices. 

We remark that by considering \textsc{$(4,4)$-nae 3sat} instead of \NAE3SAT, the bipartite graph obtained from the reduction of Yannakakis mentioned above has maximum degree at most four. Thus, by Theorem~\ref{thm:(4,4)-nae3sat}, we obtain:

\begin{theorem}[\cite{Yannakakis81a}]\label{thm:cvd-4}
\CVD\ is $\NP$-complete even when restricted to $n$-vertex bipartite graphs of maximum degree at most~$4$ and, assuming ETH, cannot be solved in $2^{o(n)}$ time.
\end{theorem}

In~\cite{HsiehLLP22}, Hsieh, Le, Le and Peng 
gave another polynomial-time reduction from \NAE3SAT\ to \CVD, which transforms an instance for \NAE3SAT\ with $n$ variables and $m$ clauses, into an equivalent instance $(G,k)$ for \CVD, where~$G$ is a subcubic bipartite graph with $6nm+30m$ vertices. Recall that we may assume (by the Sparsification Lemma) that $m=O(n)$.  
Thus, by Theorem~\ref{thm:nae3sat}, we obtain:

\begin{theorem}[\cite{HsiehLLP22}]\label{thm:cvd-3}
\CVD\ is $\NP$-complete even when restricted to subcubic $n$-vertex bipartite graphs and, assuming ETH, cannot be solved in time $2^{o(\sqrt{n})}$.
\end{theorem}
 
In this section, we will further improve Theorems~\ref{thm:cvd-4} and~\ref{thm:cvd-3} by Theorems~\ref{thm:girth} and~\ref{thm:girth-3}, respectively. We begin with the following fact.

\begin{lemma}\label{lem:3-division}
Given a graph $G$, let $G'$ be obtained from $G$ by subdividing each edge $e=xy$ in~$G$ with three new vertices $e_x, e_{xy}$ and $e_y$, thus obtaining the $5$-vertex path $xe_xe_{xy}e_yy$ in~$G'$ in which all new vertices are of degree~$2$. Assuming $G$ is triangle-free, $G$ has a cluster vertex deletion set of size at most~$k$ if and only if~$G'$ has a cluster vertex deletion set of size at most $k+m$, where $m$ is the edge number of $G$.
\end{lemma}
\begin{proof}
Observe that since $G$ is triangle-free, a cluster in $G$ is a collection of isolated vertices and edges. 

For one direction, extend a cluster vertex deletion set $S\subseteq V(G)$ to a cluster vertex deletion set $S'\subseteq V(G')$ of size $|S|+m$ as follows; see also Fig.~\ref{fig:3-division}: initially, set $S'=S$. Then, for each edge $e=xy$ in $G$,
\begin{itemize}
\item if both $x$ and $y$ are in $S$ or outside $S$, put $e_{xy}$ into $S'$;
\item if $x\in S$ and $y\notin S$, put $e_y$ into $S'$;
\item if $x\notin S$ and $y\in S$, put $e_x$ into $S'$.
\end{itemize}
To see that $G'-S'$ is $P_3$-free, notice that by construction, for each edge $e=xy$ in~$G$, exactly one of $e_x, e_{xy}$ and $e_y$ is in~$S'$, and if $e_x, e_{xy}\notin S'$ then $x\in S$, and if $e_x, x\notin S'$ then~$y\notin S$, hence~$e_{xy}\in S'$. 
Since each $P_3$ in $G'$ has the form $xe_xe_{xy}$, $e_xe_{xy}e_y$ or $e_xxe'_x$ for some edge $e=xy$ and $e'=xz$, 
it follows from these facts and the assumption that $G$ is triangle-free that $G'-S'$ is $P_3$-free. 

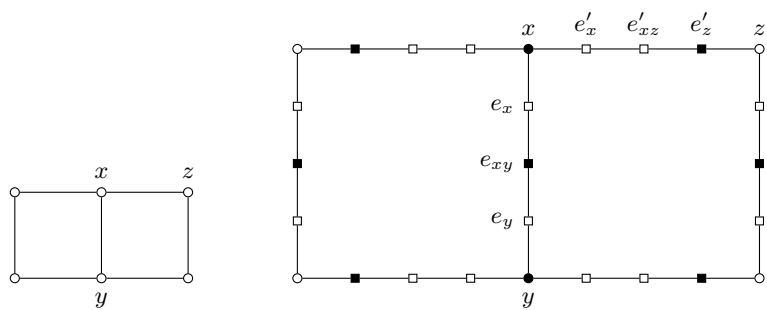
\begin{figure}[!ht]
\begin{center}
\begin{tikzpicture}[scale=.38] 
\node[vertex] (g1) at (1,1)  {}; 
\node[vertex] (g2) at (4,1)  [label=below:{\small $y$}] {};
\node[vertex] (g3) at (7,1)  {}; 
\node[vertex] (g4) at (7,4)  [label=above:{\small $z$}] {};
\node[vertex] (g5) at (4,4)  [label=above:{\small $x$}] {};
\node[vertex] (g6) at (1,4) {};
\draw (g1)--(g2)--(g3)--(g4)--(g5)--(g6)--(g1); \draw (g2)--(g5);
\end{tikzpicture}
\qquad\quad
\begin{tikzpicture}[scale=.38] 
\node[vertex] (g1) at (1,1)  {}; 
\node[vertexS] (g2) at (9,1)  [label=below:{\small $y$}] {};
\node[vertex] (g3) at (17,1)  {}; 
\node[vertex] (g4) at (17,9)  [label=above:{\small $z$}] {};
\node[vertexS] (g5) at (9,9)  [label=above:{\small $x$}] {};
\node[vertex] (g6) at (1,9) {}; 
\node[squareS] (g'12) at (3,1)  {};
\node[square] (12) at (5,1)  {};
\node[square] (g'21) at (7,1)  {}; 
\draw (g1)--(g'12)--(12)--(g'21)--(g2);
\node[square] (g'23) at (11,1)  {};
\node[square] (23) at (13,1)  {};
\node[squareS] (g'32) at (15,1)  {}; 
\draw (g2)--(g'23)--(23)--(g'32)--(g3);
\node[square] (g'34) at (17,3)  {};
\node[squareS] (34) at (17,5)  {};
\node[square] (g'43) at (17,7)  {}; 
\draw (g3)--(g'34)--(34)--(g'43)--(g4);
\node[squareS] (g'45) at (15,9)  [label=above:{\small $e'_z$}] {};
\node[square] (45) at (13,9)  [label=above:{\small $e'_{xz}$}] {};
\node[square] (g'54) at (11,9)  [label=above:{\small $e'_x$}] {};
\draw (g4)--(g'45)--(45)--(g'54)--(g5);
\node[square] (g'56) at (7,9)  {};
\node[square] (56) at (5,9)  {};
\node[squareS] (g'65) at (3,9)  {}; 
\draw (g5)--(g'56)--(56)--(g'65)--(g6);
\node[square] (g'61) at (1,7)  {};
\node[squareS] (61) at (1,5)  {};
\node[square] (g'16) at (1,3)  {}; 
\draw (g6)--(g'61)--(61)--(g'16)--(g1);
\node[square] (g'25) at (9,3)  [label=left:{\small $e_y$}] {};
\node[squareS] (25) at (9,5)  [label=left:{\small $e_{xy}$}] {};
\node[square] (g'52) at (9,7)  [label=left:{\small $e_x$}] {};
\draw (g2)--(g'25)--(25)--(g'52)--(g5);
\end{tikzpicture} 
\caption{Proof of Lemma~\ref{lem:3-division} illustrated: A triangle-free graph $G$ (left) with two highlighted edges $e=xy$ and $e'=xz$, and the graph $G'$ obtained from $G$ as described in Lemma~\ref{lem:3-division} (right); the cluster vertex deletion set $S=\{x,y\}$ of $G$ is extended to the cluster vertex deletion set $S'$ of $G'$ consisting of the nine black vertices.}\label{fig:3-division}
\end{center}
\end{figure}

For the other direction, suppose that $G'$ has a cluster vertex deletion set of size at most $k+m$, and consider such a set~$S'$ of minimum size. Then, we may assume that, for each edge $e=xy$ in $G$, $S'$ contains exactly one of $e_x, e_{xy}$ and~$e_y$: note that $e_xe_{xy}e_y$ is a~$P_3$, hence $|S'\cap\{e_x,e_{xy},e_y\}|\ge 1$, and by minimality, $|S'\cap\{e_x,e_{xy},e_y\}|\le 2$. 
Now, if $|S'\cap\{e_x,e_{xy},e_y\}| = 2$ for some edge $e=xy$ in $G$, then $S'$ can be modified to a minimum cluster vertex deletion set containing exactly one of $e_x, e_{xy}$ and~$e_y$ as follows: 
\begin{itemize}
\item suppose that $e_x, e_{xy}\in S'$. Then $x, y\not\in S'$ (if $x\in S'$ then $S'-e_x$ would be a cluster vertex deletion set of $G'$, and if $y\in S'$ then $S'-e_{xy}$ would be a cluster vertex deletion set of $G'$, contradicting the minimality of $S'$), and $S''=S'-e_{xy}+y$ is the desired cluster vertex deletion set of minimum size;
\item suppose that $e_y, e_{xy}\in S'$. Then similar to the above case, $x, y\not\in S'$, and $S''=S'-e_{xy}+x$ is the desired cluster vertex deletion set of minimum size;
\item suppose that $e_x,e_y\in S'$. Then $x,y\notin S'$ (if $x\in S'$ or $y\in S'$ then $S''=S'-e_x$, respectively $S''=S'-e_y$, would be a cluster vertex deletion set of $G'$, contradicting the minimality of $S'$), and $S''=S'-e_x+x$ is the desired cluster vertex deletion set of minimum size. 
\end{itemize}
Hence, $S=S'\cap V(G)$ has at most $k$ vertices, and $G-S$ is $P_3$-free: if there would be an induced $P_3$ $xyz$ in~$G$ with edges $e=xy$ and $e'=yz$, then, as $|S'\cap\{e_x,e_{xy},e_y\}|=1\allowbreak=|S'\cap\{e'_y,e'_{yz},e'_z\}|$, one of the $3$-paths  $xe_xe_{xy}$, $e_yye'_y$ and $e'_{yz}e'_zz$ would be outside $S'$.

Thus, $G$ has a cluster vertex deletion set of size at most $k$ if and only if $G'$ has a cluster vertex deletion set of size at most $k+m$, as claimed.
\end{proof}

We now show that, for any given tree~$T$ containing two vertices of degree~$ 3$, \CVD\ remains $\NP$-complete when restricted to $T$-free bipartite graphs of maximum degree~$4$ and with arbitrarily large girth. 
 
\begin{theorem}\label{thm:girth}
For any given integer $g\ge 3$ and any given tree $T$ containing two degree-$3$ vertices, \CVD\ is $\NP$-complete on $T$-free $n$-vertex bipartite graphs of maximum degree at most~$4$ and with girth~$>g$ and, assuming ETH,  cannot be solved in $2^{o(n)}$ time. 
\end{theorem}
\begin{proof}
Note that \CVD\ restricted to the graph class in question is in $\NP$. Below we give a polynomial-time reduction from \CVD\ restricted to bipartite graphs of degree at most~$4$ to \CVD\ restricted to $T$-free bipartite graphs of degree at most~$4$ and with arbitrarily large girth.

First, given a bipartite graph $G$ of maximum degree at most~$4$ with~$n$ vertices and $m$ edges, let $G'$ be obtained from $G$  by subdividing the edges as described in Lemma~\ref{lem:3-division}. 
Note that like $G$, $G'$ is bipartite and has maximum degree at most~$4$. 
By Lemma~\ref{lem:3-division}, $G$ has a cluster vertex deletion set of size at most~$k$ if and only if $G'$ has a cluster vertex deletion set of size at most~$k+m$. 

Now, given $g>0$ and a tree $T$ with two degree-$3$ vertices, fix an integer $t\ge \max\{\log_4 g, \allowbreak |V(T)|\}$. Then, repeating the construction in Lemma~\ref{lem:3-division} $t$ 
times, the final bipartite graph~$G'$ has girth~$4^t\cdot girth(G) > g$ and maximum degree at most~$4$, and contains no induced subgraph isomorphic to~$T$ (as the distance between two degree-3 vertices in $G'$ is larger than $|V(T)|$). 
Thus the $\NP$-hardness part of the theorem follows from the first part of Theorem~\ref{thm:cvd-4}. 
Note that $G'$ has $n+(4^t-1)m=O(n)$ vertices, 
hence, the second part of the theorem follows from the second part of Theorem~\ref{thm:cvd-4}. 
\end{proof}

Observe that if we consider subcubic bipartite graphs and make use of Theorem~\ref{thm:cvd-3} instead of Theorem~\ref{thm:cvd-4} in the proof of Theorem~\ref{thm:girth}, we obtain:

\begin{theorem}\label{thm:girth-3}
For any given integer $g\ge 3$ and any given tree $T$ containing two degree-$3$ vertices, \CVD\ is $\NP$-complete on $T$-free subcubic bipartite graphs and with girth~$>g$ and, assuming ETH,  cannot be solved in  $2^{o(\sqrt{n})}$ time. 
\end{theorem}

We now are going to show that \CCVD\ remains $\NP$-complete when restricted to bipartite graphs with arbitrarily large girth. (Notice that a reduction based on Lemma~\ref{lem:3-division}, similar to the reduction in Theorem~\ref{thm:girth}, does not work for \CCVD.) 
Let $g>0$ be a given integer. From an instance $(G,k)$ of \CVD, where $G=(X\cup Y,E)$ is a bipartite graph with girth $>g$, 
we construct an instance $(G(g),k')$, where $G(g)$ is a bipartite graph of girth $>g$,  
for \CCVD\ as follows:
\begin{itemize}
\item We may assume that $g$ is odd (otherwise, replace $g$ by $g+1$);
\item Write $X=\{x_1,x_2,\ldots,x_r\}$, $Y=\{y_1,y_2,\ldots,y_s\}$, and $n=r+s$;
\item Let $H(g,r,s)$ be the tree depicted in Fig.~\ref{fig:H(g)}; note that $H(g,r,s)$ has $6r+3gr+6s+3gs=(6+3g)n$ vertices. 
The property of $H(g,r,s)$ that will be used is that the set of all degree-3 vertices 
of $H(g, r, s)$, that is all $x_{ig}$, $1\le i \le r$, and all $y_{jg}$, $1\le j\le s$, is both an optimal cluster vertex deletion set and the unique connected cluster vertex deletion set. 
The vertices $x_{ig}$ and $y_{jg}$ will have degree~3 in the whole graph $G(g)$. 
In Fig.~\ref{fig:H(g)} the unique connected cluster vertex deletion set contains the $(g + 2)n$ black vertices.
\end{itemize}

\begin{figure}[!h] 
\begin{center}
\begin{tikzpicture}[scale=.47] 
\node[vertexS] (x10) at (3,12)  {};
\node at (3.78,12.4) {\small $x_{10}$};
\node[squareS] (a1) at (5,12) {};
\node[vertexS] (x20) at (7,12)  {};
\node at (7.78,12.4) {\small $x_{20}$};
\node[squareS] (a2) at (9,12) {};
\node[vertexS] (xi0) at (11,12) {}; 
\node[squareS] (ai) at (13,12) {};
\node[vertexS] (xr0) at (15,12)  {};
\node at (15.78,12.4) {\small $x_{r0}$};
\node[squareS] (ar) at (17,12) {};
\draw (x10)--(a1)--(x20)--(a2)--(xi0)--(ai)--(xr0)--(ar);
\node[squareS] (y10) at (3,3)  {};
\node at (3.78,2.6) {\small $y_{10}$};
\node[vertexS] (b1) at (5,3) {};
\node[squareS] (y20) at (7,3)  {};
\node at (7.78,2.6) {\small $y_{20}$};
\node[vertexS] (b2) at (9,3) {};
\node[squareS] (yj0) at (11,3) {}; 
\node[vertexS] (bj) at (13,3) {};
\node[squareS] (ys0) at (15,3)  {};
\node at (15.78,2.6) {\small $y_{s0}$};
\node[vertexS] (bs) at (17,3) {};
\draw (y10)--(b1)--(y20)--(b2)--(yj0)--(bj)--(ys0)--(bs);

\draw (ar)--(bs); 

\node[square] (1) at (3,13.2) {};
\node[vertex] (2) at (3,14) {};
\draw (x10)--(1)--(2);
\node[vertex] (1) at (5,13.2) {};
\node[square] (2) at (5,14) {};
\draw (a1)--(1)--(2);
\node[square] (1) at (7,13.2) {};
\node[vertex] (2) at (7,14) {};
\draw (x20)--(1)--(2);
\node[vertex] (1) at (9,13.2) {};
\node[square] (2) at (9,14) {};
\draw (a2)--(1)--(2);
\node[square] (1) at (11,13.2) {};
\node[vertex] (2) at (11,14) {};
\draw (xi0)--(1)--(2);
\node at (11,12.3) [label=right:{\small $\dots$}] {};
\node[vertex] (1) at (13,13.2) {};
\node[square] (2) at (13,14) {};
\draw (ai)--(1)--(2);
\node[square] (1) at (15,13.2) {};
\node[vertex] (2) at (15,14) {};
\draw (xr0)--(1)--(2);
\node[vertex] (1) at (17,13.2) {};
\node[square] (2) at (17,14) {};
\draw (ar)--(1)--(2);
\node[squareS] (x11) at (3,11)  {}; 
\node at (3.9,11) {\small $x_{11}$};
\node[vertexS] (x12) at (3,10) {}; 
\node at (3,10.2) [label=right:{\small $\vdots$}] {};
\node[squareS] (x1g) at (3,9) {}; 
\node at (3.9,9) {\small $x_{1g}$};
\draw (x10)--(x11)--(x12)--(x1g);
\node[squareS] (x21) at (7,11) {};
\node at (7.9,11) {\small $x_{21}$};
\node[vertexS] (x22) at (7,10) {};
\node at (7,10.2) [label=right:{\small $\vdots$}] {};
\node[squareS] (x2g) at (7,9) {};
\node at (7.9,9) {\small $x_{2g}$};
\draw (x20)--(x21)--(x22)--(x2g);
\node[squareS] (xi1) at (11,11) {};
\node[vertexS] (xi2) at (11,10) {};
\node at (11,10.2) [label=right:{\small $\vdots$}] {};
\node[squareS] (xig) at (11,9) {};
\draw (xi0)--(xi1)--(xi2)--(xig);
\node[squareS] (xr1) at (15,11) {};
\node at (15.9,11) {\small $x_{r1}$};
\node[vertexS] (xr2) at (15,10) {};
\node at (15,10.2) [label=right:{\small $\vdots$}] {};
\node[squareS] (xrg) at (15,9) {};
\node at (15.9,9) {\small $x_{rg}$};
\draw (xr0)--(xr1)--(xr2)--(xrg);
\node[vertex] (1) at (2,11) {};
\node[square] (2) at (1,11) {};
\draw (x11)--(1)--(2);
\node[square] (1) at (2,10) {};
\node[vertex] (2) at (1,10) {};
\draw (x12)--(1)--(2);
\node[vertex] (1) at (2,9) {};
\node[square] (2) at (1,9) {};
\draw (x1g)--(1)--(2);
\node[vertex] (1) at (6,11) {};
\node[square] (2) at (5,11) {};
\draw (x21)--(1)--(2);
\node[square] (1) at (6,10) {};
\node[vertex] (2) at (5,10) {};
\draw (x22)--(1)--(2);
\node[vertex] (1) at (6,9) {};
\node[square] (2) at (5,9) {};
\draw (x2g)--(1)--(2);
\node[vertex] (1) at (10,11) {};
\node[square] (2) at (9,11) {};
\draw (xi1)--(1)--(2);
\node[square] (1) at (10,10) {};
\node[vertex] (2) at (9,10) {};
\draw (xi2)--(1)--(2);
\node[vertex] (1) at (10,9) {};
\node[square] (2) at (9,9) {};
\draw (xig)--(1)--(2);
\node[vertex] (1) at (14,11) {};
\node[square] (2) at (13,11) {};
\draw (xr1)--(1)--(2);
\node[square] (1) at (14,10) {};
\node[vertex] (2) at (13,10) {};
\draw (xr2)--(1)--(2);
\node[vertex] (1) at (14,9) {};
\node[square] (2) at (13,9) {};
\draw (xrg)--(1)--(2);
\node[vertex] (1) at (3,1.8) {};
\node[square] (2) at (3,1) {};
\draw (y10)--(1)--(2);
\node[square] (1) at (5,1.8) {};
\node[vertex] (2) at (5,1) {};
\draw (b1)--(1)--(2);
\node[vertex] (1) at (7,1.8) {};
\node[square] (2) at (7,1) {};
\draw (y20)--(1)--(2);
\node[square] (1) at (9,1.8) {};
\node[vertex] (2) at (9,1) {};
\draw (b2)--(1)--(2);
\node[vertex] (1) at (11,1.8) {};
\node[square] (2) at (11,1) {};
\draw (yj0)--(1)--(2);
\node at (11,2.6) [label=right:{\small $\dots$}] {};
\node[square] (1) at (13,1.8) {};
\node[vertex] (2) at (13,1) {};
\draw (bj)--(1)--(2);
\node[vertex] (1) at (15,1.8) {};
\node[square] (2) at (15,1) {};
\draw (ys0)--(1)--(2);
\node[square] (1) at (17,1.8) {};
\node[vertex] (2) at (17,1) {};
\draw (bs)--(1)--(2);
\node[vertexS] (y11) at (3,4)  {}; 
\node at (3.9,4) {\small $y_{11}$};
\node[squareS] (y12) at (3,5) {};
\node at (3,5.2) [label=right:{\small $\vdots$}] {};
\node[vertexS] (y1g) at (3,6) {}; 
\node at (3.9,6) {\small $y_{1g}$};
\draw (y10)--(y11)--(y12)--(y1g);
\node[vertexS] (y21) at (7,4) {};
\node at (7.9,4) {\small $y_{21}$};
\node[squareS] (y22) at (7,5) {};
\node at (7,5.2) [label=right:{\small $\vdots$}] {};
\node[vertexS] (y2g) at (7,6) {};
\node at (7.9,6) {\small $y_{2g}$};
\draw (y20)--(y21)--(y22)--(y2g);
\node[vertexS] (yj1) at (11,4) {};
\node[squareS] (yj2) at (11,5) {};
\node at (11,5.2) [label=right:{\small $\vdots$}] {};
\node[vertexS] (yjg) at (11,6) {};
\draw (yj0)--(yj1)--(yj2)--(yjg);
\node[vertexS] (ys1) at (15,4) {};
\node at (15.9,4) {\small $y_{s1}$};
\node[squareS] (ys2) at (15,5) {};
\node at (15,5.2) [label=right:{\small $\vdots$}] {};
\node[vertexS] (ysg) at (15,6) {};
\node at (15.9,6) {\small $y_{sg}$};
\draw (ys0)--(ys1)--(ys2)--(ysg);
\node[square] (1) at (2,4) {};
\node[vertex] (2) at (1,4) {};
\draw (y11)--(1)--(2);
\node[vertex] (1) at (2,5) {};
\node[square] (2) at (1,5) {};
\draw (y12)--(1)--(2);
\node[square] (1) at (2,6) {};
\node[vertex] (2) at (1,6) {};
\draw (y1g)--(1)--(2);
\node[square] (1) at (6,4) {};
\node[vertex] (2) at (5,4) {};
\draw (y21)--(1)--(2);
\node[vertex] (1) at (6,5) {};
\node[square] (2) at (5,5) {};
\draw (y22)--(1)--(2);
\node[square] (1) at (6,6) {};
\node[vertex] (2) at (5,6) {};
\draw (y2g)--(1)--(2);
\node[square] (1) at (10,4) {};
\node[vertex] (2) at (9,4) {};
\draw (yj1)--(1)--(2);
\node[vertex] (1) at (10,5) {};
\node[square] (2) at (9,5) {};
\draw (yj2)--(1)--(2);
\node[square] (1) at (10,6) {};
\node[vertex] (2) at (9,6) {};
\draw (yjg)--(1)--(2);
\node[square] (1) at (14,4) {};
\node[vertex] (2) at (13,4) {};
\draw (ys1)--(1)--(2);
\node[vertex] (1) at (14,5) {};
\node[square] (2) at (13,5) {};
\draw (ys2)--(1)--(2);
\node[square] (1) at (14,6) {};
\node[vertex] (2) at (13,6) {};
\draw (ysg)--(1)--(2);
\end{tikzpicture} 
\caption{The tree $H(g,r,s)$. The $(g+2)n$ black vertices form an optimal (connected) cluster vertex deletion set.}\label{fig:H(g)}
\end{center}
\end{figure}
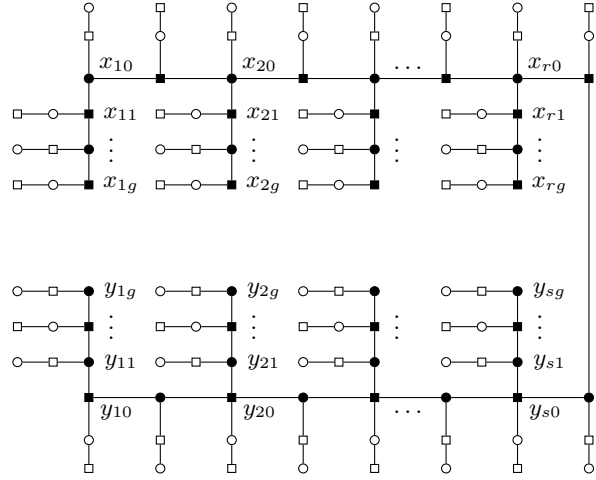

Then, let $G(g)$ be obtained from $G$ and $H(g,r,s)$ by adding an edge between $x_i$ and $x_{ig}$, $1\le i\le r$, and between $y_j$ and $y_{jg}$, $1\le j\le s$. Note that like $G$, $G(g)$ is bipartite (as $g$ is odd) and has $n'=n+(6+3g)n=(7+3g)n$ vertices. See Fig.~\ref{fig:exampleG(3)} for an example in case $g=3$.  
Finally, set $k'=k+(g+2)n$. Clearly, $(G(g),k')$ can be constructed in polynomial time from $(G,k)$.

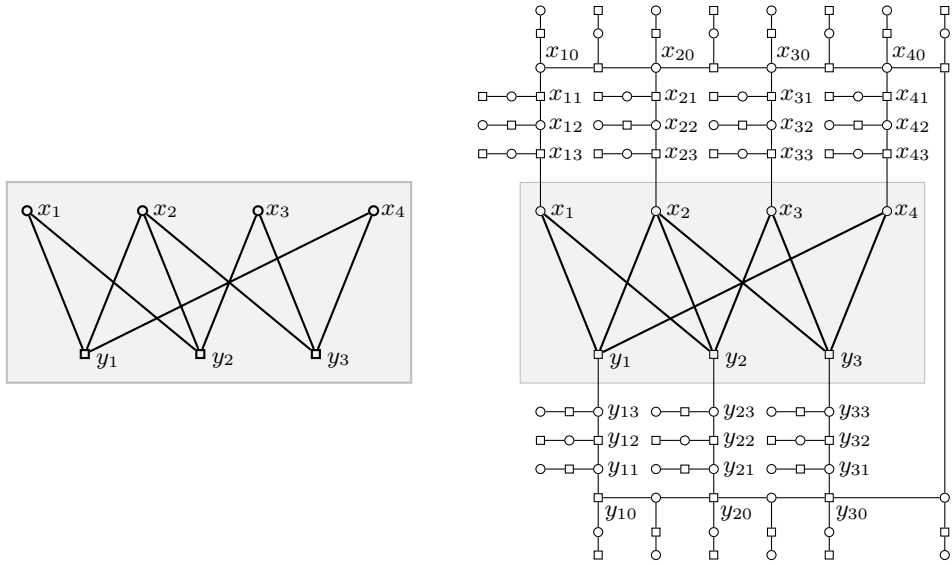
\begin{figure}[!ht]
\begin{center}
\begin{tikzpicture}[scale=.38, thick] 
\filldraw[fill=black!5!white, draw=lightgray] (2.3,6) rectangle (16.3,13);
\node[vertex] (x1) at (3,12)  {}; 
\node at (3.8,12) {\small $x_1$};
\node[vertex] (x2) at (7,12)  {}; 
\node at (7.8,12) {\small $x_2$};
\node[vertex] (x3) at (11,12) {}; 
\node at (11.7,12) {\small $x_3$};
\node[vertex] (x4) at (15,12) {}; 
\node at (15.7,12) {\small $x_4$};
\node[square] (y1) at (5,7)  {}; 
\node at (5.8,6.75) {\small $y_1$};
\node[square] (y2) at (9,7)  {}; 
\node at (9.8,6.8) {\small $y_2$};
\node[square] (y3) at (13,7) {}; 
\node at (13.8,6.8) {\small $y_3$};
\draw (x1)--(y1)--(x2)--(y2)--(x3)--(y3)--(x4)--(y1);
\draw (x1)--(y2); \draw (x2)--(y3);

\node at (5,0.12) {};
\end{tikzpicture}
\qquad
\begin{tikzpicture}[scale=.38] 
\filldraw[fill=black!5!white, draw=lightgray] (2.3,6) rectangle (16.3,13);
\node[vertex] (x1) at (3,12)  {}; 
\node at (3.8,12) {\small $x_1$};
\node[vertex] (x2) at (7,12)  {}; 
\node at (7.8,12) {\small $x_2$};
\node[vertex] (x3) at (11,12) {}; 
\node at (11.7,12) {\small $x_3$};
\node[vertex] (x4) at (15,12) {}; 
\node at (15.7,12) {\small $x_4$};
\node[square] (y1) at (5,7)  {}; 
\node at (5.8,6.75) {\small $y_1$};
\node[square] (y2) at (9,7)  {}; 
\node at (9.8,6.8) {\small $y_2$};
\node[square] (y3) at (13,7) {}; 
\node at (13.8,6.8) {\small $y_3$};
\draw[thick] (x1)--(y1)--(x2)--(y2)--(x3)--(y3)--(x4)--(y1);
\draw[thick] (x1)--(y2); \draw[thick] (x2)--(y3);
\node[vertex] (x10) at (3,17)  {};
\node at (3.78,17.5) {\small $x_{10}$};
\node[square] (a1) at (5,17) {};
\node[vertex] (x20) at (7,17)  {};
\node at (7.78,17.5) {\small $x_{20}$};
\node[square] (a2) at (9,17) {};
\node[vertex] (x30) at (11,17) {}; 
\node at (11.78,17.5) {\small $x_{30}$};
\node[square] (3) at (13,17) {};
\node[vertex] (x40) at (15,17)  {};
\node at (15.78,17.5) {\small $x_{40}$};
\node[square] (4) at (17,17) {};
\draw (x10)--(a1)--(x20)--(a2)--(x30)--(3)--(x40)--(4);
\node[square] (1) at (3,18.2) {};
\node[vertex] (2) at (3,19) {};
\draw (x10)--(1)--(2);
\node[vertex] (1) at (5,18.2) {};
\node[square] (2) at (5,19) {};
\draw (a1)--(1)--(2);
\node[square] (1) at (7,18.2) {};
\node[vertex] (2) at (7,19) {};
\draw (x20)--(1)--(2);
\node[vertex] (1) at (9,18.2) {};
\node[square] (2) at (9,19) {};
\draw (a2)--(1)--(2);
\node[square] (1) at (11,18.2) {};
\node[vertex] (2) at (11,19) {};
\draw (x30)--(1)--(2);
\node[vertex] (1) at (13,18.2) {};
\node[square] (2) at (13,19) {};
\draw (3)--(1)--(2);
\node[square] (1) at (15,18.2) {};
\node[vertex] (2) at (15,19) {};
\draw (x40)--(1)--(2);
\node[vertex] (1) at (17,18.2) {};
\node[square] (2) at (17,19) {};
\draw (4)--(1)--(2);
\node[square] (x11) at (3,16)  {}; 
\node at (3.9,16) {\small $x_{11}$};
\node[vertex] (x12) at (3,15) {};
\node at (3.9,15) {\small $x_{12}$};
\node[square] (x13) at (3,14) {}; 
\node at (3.9,14) {\small $x_{13}$};
\draw (x10)--(x11)--(x12)--(x13);
\node[square] (x21) at (7,16) {};
\node at (7.9,16) {\small $x_{21}$};
\node[vertex] (x22) at (7,15) {};
\node at (7.9,15) {\small $x_{22}$};
\node[square] (x23) at (7,14) {};
\node at (7.9,14) {\small $x_{23}$};
\draw (x20)--(x21)--(x22)--(x23);
\node[square] (x31) at (11,16) {};
\node at (11.9,16) {\small $x_{31}$};
\node[vertex] (x32) at (11,15) {};
\node at (11.9,15) {\small $x_{32}$};
\node[square] (x33) at (11,14) {};
\node at (11.9,14) {\small $x_{33}$};
\draw (x30)--(x31)--(x32)--(x33);
\node[square] (x41) at (15,16) {};
\node at (15.9,16) {\small $x_{41}$};
\node[vertex] (x42) at (15,15) {};
\node at (15.9,15) {\small $x_{42}$};
\node[square] (x43) at (15,14) {};
\node at (15.9,14) {\small $x_{43}$};
\draw (x40)--(x41)--(x42)--(x43);
\node[vertex] (1) at (2,16) {};
\node[square] (2) at (1,16) {};
\draw (x11)--(1)--(2);
\node[square] (1) at (2,15) {};
\node[vertex] (2) at (1,15) {};
\draw (x12)--(1)--(2);
\node[vertex] (1) at (2,14) {};
\node[square] (2) at (1,14) {};
\draw (x13)--(1)--(2);
\node[vertex] (1) at (6,16) {};
\node[square] (2) at (5,16) {};
\draw (x21)--(1)--(2);
\node[square] (1) at (6,15) {};
\node[vertex] (2) at (5,15) {};
\draw (x22)--(1)--(2);
\node[vertex] (1) at (6,14) {};
\node[square] (2) at (5,14) {};
\draw (x23)--(1)--(2);
\node[vertex] (1) at (10,16) {};
\node[square] (2) at (9,16) {};
\draw (x31)--(1)--(2);
\node[square] (1) at (10,15) {};
\node[vertex] (2) at (9,15) {};
\draw (x32)--(1)--(2);
\node[vertex] (1) at (10,14) {};
\node[square] (2) at (9,14) {};
\draw (x33)--(1)--(2);
\node[vertex] (1) at (14,16) {};
\node[square] (2) at (13,16) {};
\draw (x41)--(1)--(2);
\node[square] (1) at (14,15) {};
\node[vertex] (2) at (13,15) {};
\draw (x42)--(1)--(2);
\node[vertex] (1) at (14,14) {};
\node[square] (2) at (13,14) {};
\draw (x43)--(1)--(2);
\draw (x13)--(x1); \draw (x23)--(x2); \draw (x33)--(x3); \draw (x43)--(x4);
\node[square] (y10) at (5,2)  {};
\node at (5.8,1.45) {\small $y_{10}$};
\node[vertex] (b1) at (7,2) {};
\node[square] (y20) at (9,2)  {};
\node at (9.8,1.45) {\small $y_{20}$};
\node[vertex] (b2) at (11,2) {};
\node[square] (y30) at (13,2) {}; 
\node at (13.8,1.45) {\small $y_{30}$};
\node[vertex] (3) at (17,2) {};
\draw (y10)--(b1)--(y20)--(b2)--(y30)--(3);

\draw (4)--(3); 

\node[vertex] (1) at (5,.8) {};
\node[square] (2) at (5,0) {};
\draw (y10)--(1)--(2);
\node[square] (1) at (7,.8) {};
\node[vertex] (2) at (7,0) {};
\draw (b1)--(1)--(2);
\node[vertex] (1) at (9,.8) {};
\node[square] (2) at (9,0) {};
\draw (y20)--(1)--(2);
\node[square] (1) at (11,.8) {};
\node[vertex] (2) at (11,0) {};
\draw (b2)--(1)--(2);
\node[vertex] (1) at (13,.8) {};
\node[square] (2) at (13,0) {};
\draw (y30)--(1)--(2);
\node[square] (1) at (17,.8) {};
\node[vertex] (2) at (17,0) {};
\draw (3)--(1)--(2);
\node[vertex] (y11) at (5,3)  {};
\node at (5.9,3) {\small $y_{11}$};
\node[square] (y12) at (5,4) {};
\node at (5.9,4) {\small $y_{12}$};
\node[vertex] (y13) at (5,5) {};
\node at (5.9,5) {\small $y_{13}$};
\draw (y10)--(y11)--(y12)--(y13);
\node[vertex] (y21) at (9,3) {};
\node at (9.9,3) {\small $y_{21}$};
\node[square] (y22) at (9,4) {};
\node at (9.9,4) {\small $y_{22}$};
\node[vertex] (y23) at (9,5) {};
\node at (9.9,5) {\small $y_{23}$};
\draw (y20)--(y21)--(y22)--(y23);
\node[vertex] (y31) at (13,3) {};
\node at (13.9,3) {\small $y_{31}$};
\node[square] (y32) at (13,4) {};
\node at (13.9,4) {\small $y_{32}$};
\node[vertex] (y33) at (13,5) {};
\node at (13.9,5) {\small $y_{33}$};
\draw (y30)--(y31)--(y32)--(y33);
\node[square] (1) at (4,3) {};
\node[vertex] (2) at (3,3) {};
\draw (y11)--(1)--(2);
\node[vertex] (1) at (4,4) {};
\node[square] (2) at (3,4) {};
\draw (y12)--(1)--(2);
\node[square] (1) at (4,5) {};
\node[vertex] (2) at (3,5) {};
\draw (y13)--(1)--(2);
\node[square] (1) at (8,3) {};
\node[vertex] (2) at (7,3) {};
\draw (y21)--(1)--(2);
\node[vertex] (1) at (8,4) {};
\node[square] (2) at (7,4) {};
\draw (y22)--(1)--(2);
\node[square] (1) at (8,5) {};
\node[vertex] (2) at (7,5) {};
\draw (y23)--(1)--(2);
\node[square] (1) at (12,3) {};
\node[vertex] (2) at (11,3) {};
\draw (y31)--(1)--(2);
\node[vertex] (1) at (12,4) {};
\node[square] (2) at (11,4) {};
\draw (y32)--(1)--(2);
\node[square] (1) at (12,5) {};
\node[vertex] (2) at (11,5) {};
\draw (y33)--(1)--(2);
\draw (y13)--(y1); \draw (y23)--(y2); \draw (y33)--(y3);
\end{tikzpicture} 
\caption{An example of the reduction from \CVD\ to \CCVD: A bipartite graph $G$ (left) and the bipartite graph $G(3)$ (right) obtained from $G$ and $H(3,4,3)$; the bipartition of the vertex set is indicated by circle and rectangle vertices.}\label{fig:exampleG(3)}
\end{center}
\end{figure}

Now, let $S$ be a cluster vertex deletion set of $G$ of size at most $k$. Then $G(g)$ has a connected cluster vertex deletion set $S'$ of size $|S|+(g+2)n\le k'$: $S'$ is obtained from $S$ by adding all vertices of $H(g,r,s)$ with degree $3$ in $G(g)$ (the $(g+2)n$ black vertices in Fig.~\ref{fig:H(g)}).  
Observe that $S'$ induces a connected subgraph in $G(g)$ since every vertex in $S$ is adjacent to some $x_{ig}$ or $y_{jg}$, and all vertices of $H(g,r,s)$ with degree $3$ in $G(g)$ induce a connected subgraph in $G(g)$.
 
Conversely, let $S'$ be a (connected or not) cluster vertex deletion set of $G(g)$ of size at most $k'$. Since every vertex $u$ in $H(g,r,s)$ with degree $3$ in $G(g)$ (the black vertices in Fig.~\ref{fig:H(g)}) belongs to an induced $P_3=uvw$ in $H(g,r,s)$ with $\deg_{G(g)}(v)=2$ and $\deg_{G(g)}(w)=1$, we may assume that $S'$ contains all $(g+2)n$ vertices of $H(g,r,s)$ with degree $3$ (and no other vertices of $H(g,r,s)$). Let~$S$ be the restriction of $S'$ on $V(G)$. Then $S$ is a cluster vertex deletion set of $G$ of size $|S|=|S'|-(g+2)n\le k$. 

Observe that the girth of $G(g)$ is at least $\max\{girth(G), 2g+6\}>g$ and the maximum degree of $G(g)$ is one more than the maximum degree of $G$. Hence, by Theorems~\ref{thm:girth} and~\ref{thm:girth-3}, we obtain:

\begin{theorem}\label{thm:girth-ccvd}
For any given integer $g\ge 3$, \CCVD\ is $\NP$-complete on bipartite graphs of maximum degree at most~$5$ and with girth~$>g$ and, assuming ETH,  cannot be solved in $2^{o(n)}$ time.
\end{theorem} 

\begin{theorem}\label{thm:girth-ccvd-4}
For any given integer $g\ge 3$, \CCVD\ is $\NP$-complete on bipartite graphs of maximum degree at most~$4$ and with girth~$>g$ and, assuming ETH,  cannot be solved in $2^{o(\sqrt{n})}$ time.
\end{theorem}

\section{$H$-free graphs: $\NP$-completeness cases}\label{sec:NPc-part}
In this section we give the proof of the $\NP$-completeness part of Theorems~\ref{thm:H-free} and~\ref{thm:H-free-connected}. 

Let $H$ be a fixed graph. By Proposition~\ref{pro:P4}, \CVD\ is polynomially solvable on $H$-free graphs whenever $H$ is an induced subgraph of the $4$-vertex path $P_4$. The following fact is easy to see:

\begin{observation}\label{obs:P4}
A graph is an induced subgraph of the $4$-path $P_4$ if and only if it is a $\{3P_1, 2P_2\}$-free forest.
\end{observation}

Thus, it remains to consider the cases where $H$ contains a cycle or a $3P_1$ or a $2P_2$ as an induced subgraph.
 
Now, if $H$ contains a cycle then graphs of girth $> g=|V(H)|$ are $H$-free, hence Theorems~\ref{thm:girth} and~\ref{thm:girth-ccvd} imply that \CVD\ and \CCVD\ are $\NP$-complete on $H$-free graphs and, assuming ETH, cannot be solved in $2^{o(n)}$ time on $H$-free $n$-vertex graphs. 
If $H$ contains a $3P_1$ or a $2P_2$ then $\{3P_1, 2P_2\}$-free graphs are $H$-free graphs, hence Theorems~\ref{thm:3P1and2P2} and~\ref{thm:3P1and2P2-ccvd} imply that \CVD\ and \CCVD\ are $\NP$-complete on $H$-free graphs and, assuming ETH, cannot be solved in $2^{o(n)}$ time on $H$-free $n$-vertex graphs. 

The proofs of Theorems~\ref{thm:H-free} and~\ref{thm:H-free-connected} are complete.

\section{Conclusion}
We have found a complete characterization of graphs $H$ for which \CVD\ on $H$-free graphs is polynomially solvable and for which it is $\NP$-complete (Theorem~\ref{thm:H-free}). The same complexity dichotomy holds also for \CCVD\ (Theorem~\ref{thm:H-free-connected}). 

We remark that a complexity dichotomy for \VC\ and \CVC\ on $H$-free graphs, like Theorem~\ref{thm:H-free} and Theorem~\ref{thm:H-free-connected} for \CVD\ and \CCVD, respectively, seems very hard to achieve. 
Indeed, it is a long-standing open problem whether there exists a constant $t$ for which \VC\ or \CVC\ is $\NP$-complete on $P_t$-free graphs. 
So far it is known that such a constant $t$, if any, must be at least $7$ for \VC~\cite{GrzesikKPP22}, respectively, at least~$6$ for \CVC~\cite{JohnsonPP20}.  

Let $\cal H$ be a set of (possibly infinitely many) graphs. A natural question generalizing the case of one forbidden induced subgraph is: what is the complexity of \CVD\ and of \CCVD\ on ${\cal H}$-free graphs? 
The case ${\cal H}=\{H\}$ is completely solved by Theorems~\ref{thm:H-free} and~\ref{thm:H-free-connected}. The case ${\cal H}=\{C_\ell\mid \ell\ge 4\}$, also known as chordal graphs, addressed in~\cite{0001KOY18} is still open. 
The next step may be the case of two-element sets ${\cal H} =\{H_1,H_2\}$; 
in particular, ${\cal H} =\{H,\overline{H}\}$. 
Another interesting problem is to clear the complexity of \CVD\ and \CCVD\ on line graphs, a well-studied graph class defined by excluding nine small induced subgraphs.


\bibliography{H-freeCVD}

\begin{thebibliography}{10}

\bibitem{AprileDFH22}
Manuel Aprile, Matthew Drescher, Samuel Fiorini, and Tony Huynh.
\newblock A tight approximation algorithm for the cluster vertex deletion
  problem.
\newblock {\em Math. Program.}, 197(2):1069--1091, 2023.
\newblock \href {https://doi.org/10.1007/s10107-021-01744-w}
  {\path{doi:10.1007/s10107-021-01744-w}}.

\bibitem{BoralCKP16}
Anudhyan Boral, Marek Cygan, Tomasz Kociumaka, and Marcin Pilipczuk.
\newblock A {F}ast {B}ranching {A}lgorithm for {C}luster {V}ertex {D}eletion.
\newblock {\em Theory Comput. Syst.}, 58(2):357--376, 2016.
\newblock \href {https://doi.org/10.1007/s00224-015-9631-7}
  {\path{doi:10.1007/s00224-015-9631-7}}.

\bibitem{0001KOY18}
Yixin Cao, Yuping Ke, Yota Otachi, and Jie You.
\newblock Vertex deletion problems on chordal graphs.
\newblock {\em Theor. Comput. Sci.}, 745:75--86, 2018.
\newblock \href {https://doi.org/10.1016/j.tcs.2018.05.039}
  {\path{doi:10.1016/j.tcs.2018.05.039}}.

\bibitem{ChakrabortyCPP21}
Dibyayan Chakraborty, L.~Sunil Chandran, Sajith Padinhatteeri, and Raji~R.
  Pillai.
\newblock Algorithms and {C}omplexity of $s$-{C}lub {C}luster {V}ertex
  {D}eletion.
\newblock In Paola Flocchini and Lucia Moura, editors, {\em Combinatorial
  Algorithms - 32nd International Workshop, {IWOCA} 2021, Ottawa, ON, Canada,
  Proceedings}, volume 12757 of {\em Lecture Notes in Computer Science}, pages
  152--164. Springer, 2021.
\newblock \href {https://doi.org/10.1007/978-3-030-79987-8\_11}
  {\path{doi:10.1007/978-3-030-79987-8\_11}}.

\bibitem{ChudnovskyCLSV05}
Maria Chudnovsky, G{\'{e}}rard Cornu{\'{e}}jols, Xinming Liu, Paul~D. Seymour,
  and Kristina Vuskovic.
\newblock Recognizing {B}erge {G}raphs.
\newblock {\em Combinatorica}, 25(2):143--186, 2005.
\newblock \href {https://doi.org/10.1007/s00493-005-0012-8}
  {\path{doi:10.1007/s00493-005-0012-8}}.

\bibitem{CorneilLB81}
Derek~G. Corneil, H.~Lerchs, and L.~Stewart Burlingham.
\newblock Complement reducible graphs.
\newblock {\em Discret. Appl. Math.}, 3(3):163--174, 1981.
\newblock \href {https://doi.org/10.1016/0166-218X(81)90013-5}
  {\path{doi:10.1016/0166-218X(81)90013-5}}.

\bibitem{CorneilPS85}
Derek~G. Corneil, Yehoshua Perl, and Lorna~K. Stewart.
\newblock A {L}inear {R}ecognition {A}lgorithm for {C}ographs.
\newblock {\em {SIAM} J. Comput.}, 14(4):926--934, 1985.
\newblock \href {https://doi.org/10.1137/0214065} {\path{doi:10.1137/0214065}}.

\bibitem{CourcelleER93}
Bruno Courcelle, Joost Engelfriet, and Grzegorz Rozenberg.
\newblock Handle-{R}ewriting {H}ypergraph {G}rammars.
\newblock {\em J. Comput. Syst. Sci.}, 46(2):218--270, 1993.
\newblock \href {https://doi.org/10.1016/0022-0000(93)90004-G}
  {\path{doi:10.1016/0022-0000(93)90004-G}}.

\bibitem{CourcelleMR00}
Bruno Courcelle, Johann~A. Makowsky, and Udi Rotics.
\newblock Linear {T}ime {S}olvable {O}ptimization {P}roblems on {G}raphs of
  {B}ounded {C}lique-{W}idth.
\newblock {\em Theory Comput. Syst.}, 33(2):125--150, 2000.
\newblock \href {https://doi.org/10.1007/s002249910009}
  {\path{doi:10.1007/s002249910009}}.

\bibitem{GartlandL20}
Peter Gartland and Daniel Lokshtanov.
\newblock Independent {S}et on ${P}_k$-{F}ree {G}raphs in {Q}uasi-{P}olynomial
  {T}ime.
\newblock In Sandy Irani, editor, {\em 61st {IEEE} Annual Symposium on
  Foundations of Computer Science, {FOCS} 2020, Durham, NC, USA}, pages
  613--624. {IEEE}, 2020.
\newblock \href {https://doi.org/10.1109/FOCS46700.2020.00063}
  {\path{doi:10.1109/FOCS46700.2020.00063}}.

\bibitem{GolovachJPS17}
Petr~A. Golovach, Matthew Johnson, Dani{\"{e}}l Paulusma, and Jian Song.
\newblock A {S}urvey on the {C}omputational {C}omplexity of {C}oloring {G}raphs
  with {F}orbidden subgraphs.
\newblock {\em J. Graph Theory}, 84(4):331--363, 2017.
\newblock \href {https://doi.org/10.1002/jgt.22028}
  {\path{doi:10.1002/jgt.22028}}.

\bibitem{GLS1988}
Martin Gr{\"{o}}tschel, L{\'{a}}szl{\'{o}} Lov{\'{a}}sz, and Alexander
  Schrijver.
\newblock {\em Geometric {A}lgorithms and {C}ombinatorial {O}ptimization}.
\newblock Springer, 1988.
\newblock \href {https://doi.org/10.1007/978-3-642-97881-4}
  {\path{doi:10.1007/978-3-642-97881-4}}.

\bibitem{GrzesikKPP22}
Andrzej Grzesik, Tereza Klimosov{\'{a}}, Marcin Pilipczuk, and Michal
  Pilipczuk.
\newblock Polynomial-time {A}lgorithm for {M}aximum {W}eight {I}ndependent
  {S}et on ${P}_6$-free {G}raphs.
\newblock {\em {ACM} Trans. Algorithms}, 18(1):4:1--4:57, 2022.
\newblock \href {https://doi.org/10.1145/3414473} {\path{doi:10.1145/3414473}}.

\bibitem{HsiehLLP22}
Sun{-}Yuan Hsieh, Ho{\`{a}}ng{-}Oanh Le, Van~Bang Le, and Sheng{-}Lung Peng.
\newblock On the $d$-{C}law {V}ertex {D}eletion {P}roblem.
\newblock {\em Algorithmica}, 2023.
\newblock \href {https://doi.org/10.1007/s00453-023-01144-w}
  {\path{doi:10.1007/s00453-023-01144-w}}.

\bibitem{HuffnerKMN10}
Falk H{\"{u}}ffner, Christian Komusiewicz, Hannes Moser, and Rolf Niedermeier.
\newblock Fixed-{P}arameter {A}lgorithms for {C}luster {V}ertex {D}eletion.
\newblock {\em Theory Comput. Syst.}, 47(1):196--217, 2010.
\newblock \href {https://doi.org/10.1007/s00224-008-9150-x}
  {\path{doi:10.1007/s00224-008-9150-x}}.

\bibitem{ImpagliazzoP01}
Russell Impagliazzo and Ramamohan Paturi.
\newblock On the {C}omplexity of $k$-{SAT}.
\newblock {\em J. Comput. Syst. Sci.}, 62(2):367--375, 2001.
\newblock \href {https://doi.org/10.1006/jcss.2000.1727}
  {\path{doi:10.1006/jcss.2000.1727}}.

\bibitem{ImpagliazzoPZ01}
Russell Impagliazzo, Ramamohan Paturi, and Francis Zane.
\newblock Which {P}roblems {H}ave {S}trongly {E}xponential {C}omplexity?
\newblock {\em J. Comput. Syst. Sci.}, 63(4):512--530, 2001.
\newblock \href {https://doi.org/10.1006/jcss.2001.1774}
  {\path{doi:10.1006/jcss.2001.1774}}.

\bibitem{JohnsonS99}
David~S. Johnson and Mario Szegedy.
\newblock What are the {L}east {T}ractable {I}nstances of {M}ax {I}ndependent
  {S}et?
\newblock In Robert~Endre Tarjan and Tandy~J. Warnow, editors, {\em Proceedings
  of the Tenth Annual {ACM-SIAM} Symposium on Discrete Algorithms, Baltimore,
  Maryland, {USA}}, pages 927--928. {ACM/SIAM}, 1999.
\newblock URL: \url{http://dl.acm.org/citation.cfm?id=314500.315093}.

\bibitem{JohnsonPP20}
Matthew Johnson, Giacomo Paesani, and Dani{\"{e}}l Paulusma.
\newblock Connected {V}ertex {C}over for $(s{P}_1+{P}_5)$-{F}ree {G}raphs.
\newblock {\em Algorithmica}, 82(1):20--40, 2020.
\newblock \href {https://doi.org/10.1007/s00453-019-00601-9}
  {\path{doi:10.1007/s00453-019-00601-9}}.

\bibitem{Kaminski12}
Marcin Kamin\'ski.
\newblock {M}ax-{C}ut and containment relations in graphs.
\newblock {\em Theor. Comput. Sci.}, 438:89--95, 2012.
\newblock \href {https://doi.org/10.1016/j.tcs.2012.02.036}
  {\path{doi:10.1016/j.tcs.2012.02.036}}.

\bibitem{Komusiewicz18}
Christian Komusiewicz.
\newblock Tight {R}unning {T}ime {L}ower {B}ounds for {V}ertex {D}eletion
  {P}roblems.
\newblock {\em {ACM} Trans. Comput. Theory}, 10(2):6:1--6:18, 2018.
\newblock \href {https://doi.org/10.1145/3186589} {\path{doi:10.1145/3186589}}.

\bibitem{Korobitsin92}
D.V. Korobitsin.
\newblock On the complexity of domination number determination in monogenic
  classes of graphs.
\newblock {\em Discrete Math. Appl.}, 2:191--200, 1992.
\newblock \href {https://doi.org/10.1515/dma.1992.2.2.191}
  {\path{doi:10.1515/dma.1992.2.2.191}}.

\bibitem{KralKTW01}
Daniel Kr{\'{a}}l, Jan Kratochv{\'{\i}}l, Zsolt Tuza, and Gerhard~J. Woeginger.
\newblock Complexity of {C}oloring {G}raphs without {F}orbidden {I}nduced
  {S}ubgraphs.
\newblock In Andreas Brandst{\"{a}}dt and Van~Bang Le, editors, {\em
  Graph-Theoretic Concepts in Computer Science, 27th International Workshop,
  {WG} 2001, Boltenhagen, Germany, Proceedings}, volume 2204 of {\em Lecture
  Notes in Computer Science}, pages 254--262. Springer, 2001.
\newblock \href {https://doi.org/10.1007/3-540-45477-2\_23}
  {\path{doi:10.1007/3-540-45477-2\_23}}.

\bibitem{LeLe22}
Hoang-Oanh Le and Van~Bang Le.
\newblock {Complexity of the {C}luster {V}ertex {D}eletion {P}roblem on
  ${H}$-{F}ree {G}raphs}.
\newblock In Stefan Szeider, Robert Ganian, and Alexandra Silva, editors, {\em
  47th International Symposium on Mathematical Foundations of Computer Science
  (MFCS 2022)}, volume 241 of {\em Leibniz International Proceedings in
  Informatics (LIPIcs)}, pages 68:1--68:10, Dagstuhl, Germany, 2022. Schloss
  Dagstuhl -- Leibniz-Zentrum f{\"u}r Informatik.
\newblock \href {https://doi.org/10.4230/LIPIcs.MFCS.2022.68}
  {\path{doi:10.4230/LIPIcs.MFCS.2022.68}}.

\bibitem{LewisY80}
John~M. Lewis and Mihalis Yannakakis.
\newblock The {N}ode-{D}eletion {P}roblem for {H}ereditary {P}roperties is
  {NP}-complete.
\newblock {\em J. Comput. Syst. Sci.}, 20(2):219--230, 1980.
\newblock \href {https://doi.org/10.1016/0022-0000(80)90060-4}
  {\path{doi:10.1016/0022-0000(80)90060-4}}.

\bibitem{Moret98}
Bernard M.~E. Moret.
\newblock {\em Theory of {C}omputation}.
\newblock Addison-Wesley-Longman, 1998.

\bibitem{Munaro17}
Andrea Munaro.
\newblock Boundary classes for graph problems involving non-local properties.
\newblock {\em Theor. Comput. Sci.}, 692:46--71, 2017.
\newblock \href {https://doi.org/10.1016/j.tcs.2017.06.012}
  {\path{doi:10.1016/j.tcs.2017.06.012}}.

\bibitem{Murphy92}
Owen~J. Murphy.
\newblock Computing independent sets in graphs with large girth.
\newblock {\em Discret. Appl. Math.}, 35(2):167--170, 1992.
\newblock \href {https://doi.org/10.1016/0166-218X(92)90041-8}
  {\path{doi:10.1016/0166-218X(92)90041-8}}.

\bibitem{SauS21}
Ignasi Sau and U{\'{e}}verton dos Santos~Souza.
\newblock Hitting forbidden induced subgraphs on bounded treewidth graphs.
\newblock {\em Inf. Comput.}, 281:104812, 2021.
\newblock \href {https://doi.org/10.1016/j.ic.2021.104812}
  {\path{doi:10.1016/j.ic.2021.104812}}.

\bibitem{Tovey84}
Craig~A. Tovey.
\newblock A simplified {NP}-complete satisfiability problem.
\newblock {\em Discret. Appl. Math.}, 8(1):85--89, 1984.
\newblock \href {https://doi.org/10.1016/0166-218X(84)90081-7}
  {\path{doi:10.1016/0166-218X(84)90081-7}}.

\bibitem{Tsur21}
Dekel Tsur.
\newblock Faster {P}arameterized {A}lgorithm for {C}luster {V}ertex {D}eletion.
\newblock {\em Theory Comput. Syst.}, 65(2):323--343, 2021.
\newblock \href {https://doi.org/10.1007/s00224-020-10005-w}
  {\path{doi:10.1007/s00224-020-10005-w}}.

\bibitem{Yannakakis78}
Mihalis Yannakakis.
\newblock Node- and {E}dge-{D}eletion {NP}-{C}omplete {P}roblems.
\newblock In Richard~J. Lipton, Walter~A. Burkhard, Walter~J. Savitch, Emily~P.
  Friedman, and Alfred~V. Aho, editors, {\em Proceedings of the 10th Annual
  {ACM} Symposium on Theory of Computing, San Diego, California, {USA}}, pages
  253--264. {ACM}, 1978.
\newblock \href {https://doi.org/10.1145/800133.804355}
  {\path{doi:10.1145/800133.804355}}.

\bibitem{Yannakakis81a}
Mihalis Yannakakis.
\newblock Node-{D}eletion {P}roblems on {B}ipartite {G}raphs.
\newblock {\em {SIAM} J. Comput.}, 10(2):310--327, 1981.
\newblock \href {https://doi.org/10.1137/0210022} {\path{doi:10.1137/0210022}}.

\end{thebibliography}

\appendix

\section{Computing the cluster vertex deletion number of cographs using the cotrees}\label{app:P4free}

Recall that $P_4$-free graphs are also called \emph{cographs}~\cite{CorneilLB81}. 
More precisely, for vertex-disjoint graphs $G_i=(V_i,E_i)$, $i=1,2$, let $G_1\cojoin G_2$ be the union (or \emph{co-join}) of $G_1$ and~$G_2$, 
\[G_1\cojoin G_2=(V_1\cup V_2, E_1\cup E_2),
\]
and let $G_1\join G_2$ be the \emph{join} of $G_1$ and $G_2$, 
\[G_1\join G_2=(V_1\cup V_2, E_1\cup E_2\cup\{uv\mid u\in V_1, v\in V_2\}).
\]
With these notations, cographs are exactly those graphs that can be constructed from the one-vertex graph by applying the join and co-join operations. Thus, a cograph is the one-vertex graph or is the join of two smaller cographs or is the co-join of two smaller cographs.

Recall that $S\subseteq V(G)$ is a vertex cover if $G-S$ is edgeless and is a cluster vertex deletion set if $G-S$ is a cluster graph. Let $\tau(G)$ and $\varsigma(G)$ denote the vertex cover number and the cluster vertex deletion number of $G$, respectively,  
\begin{align*}
\tau(G)&=\min\{|S| : \text{$S$ is a vertex cover of $G$}\},\\ 
\varsigma(G)&=\min\{|S| :  \text{$S$ is a cluster vertex deletion set of $G$}\}. 
\end{align*}

We will see that $\tau(G)$ and $\varsigma(G)$ can be computed efficiently when restricted to cographs. The calculation is based on the following fact:

\begin{lemma}\label{lem:join-cojoin}
For any (not necessarily $P_4$-free) graphs $G_1$ and $G_2$, the following relations hold:
\begin{align}
\tau(G_1\cojoin G_2)& =\tau(G_1)+\tau(G_2)\label{tau-cojoin};\\
\tau(G_1\join G_2) & = \min\{\tau(G_1)+|V(G_2)|, \tau(G_2)+|V(G_1)|\};\label{tau-join}\\ 
\varsigma(G_1\cojoin G_2)& =\varsigma(G_1)+\varsigma(G_2);\label{sigma-cojoin}\\
\varsigma(G_1\join G_2) & = \min\{\varsigma(G_1)+|V(G_2)|, \varsigma(G_2)+|V(G_1)|, \tau(\overline{G_1})+\tau(\overline{G_2})\}.\label{sigma-join}
\end{align}
\end{lemma}
\begin{proof}
(\ref{tau-cojoin}) and (\ref{sigma-cojoin}) are trivial.

\smallskip\noindent
(\ref{tau-join}): Let $S_i$ be a vertex cover of $G_i$ of optimal size $\tau(G_i)$, $i=1,2$. 
Then $S_1\cup V(G_2)$ and $S_2\cup V(G_1)$ are vertex covers of $G_1\join G_2$. 
Hence $\tau(G_1\join G_2)\le\min\{|S_1|+|V(G_2)|,|S_2|+|V(G_1)|\}=\min\{\tau(G_1)+|V(G_2)|,\tau(G_2)+|V(G_1)|\}$. 

For the other direction, let $S$ be  a vertex cover of $G_1\join G_2$ of optimal size, and write $S_i=S\cap V(G_i)$. Then $S_i$ is a vertex cover of $G_i$, and moreover, $S_1=V(G_1)$ or else $S_2=V(G_2)$ (because $S_i=V(G_i)$ for some $i$ is needed to cover the edges between $G_1$ and $G_2$). 
Hence $\tau(G_1\join G_2)\ge\min\{|S_1|+|V(G_2)|,|S_2|+|V(G_1)|\}\ge \min\{\tau(G_1)+|V(G_2)|,\tau(G_2)+|V(G_1)|\}$.

\smallskip\noindent
(\ref{sigma-join}): Let $S_i$ be a cluster vertex deletion set of $G_i$ of optimal size $\varsigma(G_i)$, $i=1,2$. 
Then $S_1\cup V(G_2)$ and $S_2\cup V(G_1)$ are cluster vertex deletion sets of $G_1\join G_2$. 
Hence $\varsigma(G_1\join G_2)\le\min\{|S_1|+|V(G_2)|,|S_2|+|V(G_1)|\}=\min\{\varsigma(G_1)+|V(G_2)|,\varsigma(G_2)+|V(G_1)|\}$. 
Let $S_i$ be a vertex cover of $\overline{G_i}$ of optimal size $\tau(\overline{G_i})$, $i=1,2$. Then $S_1\cup S_2$ is a cluster vertex deletion set of $G_1\join G_2$, hence  $\varsigma(G_1\join G_2)\le |S_1|+|S_2|=\tau(\overline{G_1})+\tau(\overline{G_2})$.

For the other direction, let $S$ be  a cluster vertex deletion set of $G_1\join G_2$ of optimal size, and write $S_i=S\cap V(G_i)$. Then $S_i$ is a cluster vertex deletion set of $G_i$, and moreover,
\begin{itemize}
\item if $G_1-S_1$ is not a clique then $S_2=V(G_2)$, likewise
\item if $G_2-S_2$ is not a clique then $S_1=V(G_1)$.
\end{itemize}
In these two cases, $|S|=\varsigma(G_1\join G_2)\ge \min\{|S_1|+|V(G_2)|, |S_2|+|V(G_1)|\}\ge \min\{\varsigma(G_1)+|V(G_2)|,\varsigma(G_2)+|V(G_1)|\}$.
In the third case where each of $G_1-S_1$ and $G_2-S_2$ is a clique, $S_1$ and $S_2$ are vertex covers of $\overline{G_1}$ and $\overline{G_2}$, respectively.
Hence in this case, $|S|=\varsigma(G_1\join G_2)=|S_1|+|S_2|\ge \tau(\overline{G_1})+\tau(\overline{G_2})$.
\end{proof}

\begin{remark}\label{rem:1}
For any integer $r\ge2$, Lemma~\ref{lem:join-cojoin} holds accordingly for $G_1\cojoin G_2\cojoin\cdots\allowbreak\cojoin G_r=\allowbreak G_1\cojoin(G_2\cojoin\cdots\cojoin G_r)$ and $G_1\join G_2\join\cdots\join G_r=G_1\join(G_2\join\cdots\join G_r)$. We also note that Lemma~\ref{lem:join-cojoin} holds for the weighted version, too.
\end{remark}

With each cograph $G=(V,E)$, one can associate a so-called \emph{cotree} $T$ of~$G$ as follows.
\begin{itemize}
\item The leaves of $T$ are the vertices of $G$;
\item Every internal node of $T$ has a label~$\cojoin$ or~$\join$, and has at least two children;
\item No two internal nodes of $T$ with the same label are adjacent;
\item Two vertices $u$ and $v$ of $G$ are (non-)adjacent if and only if the least common ancestor of~$u$ and~$v$ in~$T$ has label~$\join$ (respectively, $\cojoin$). 
\end{itemize}
In particular, the cotree of an $n$-vertex cograph has at most $2n-1$ nodes. 

Note that, for any internal node $v$ of $T$, the subtree $T_v$ of $T$ rooted at $v$ is the cotree of the subgraph of $G$ induced by the leaves of ${T}_v$. 
The cograph corresponding to ${T}_v$ where $v$ has label $\cojoin$ is the disjoint union of the cographs corresponding to the children of $v$. 
The cograph corresponding to ${T}_v$ where $v$ has label $\join$ is the join of the cographs corresponding to the children of $v$. 

In particular, the cotree of $\overline{G}$ can be obtained from the cotree of $G$ by changing the label~$\cojoin$ to~$\join$ and~$\join$ to~$\cojoin$.
 
In~\cite{CorneilPS85}, a linear time algorithm is given for recognizing if a given graph is a cograph, and if so, constructing its cotree. 
Note that the cotree can immediately  be transformed to an equivalent binary tree; see Fig.~\ref{fig:cotree} for an example of a cograph $G$, the cotree of $G$ and its binary version. 
For simplification, we will use the binary cotree in our algorithm below.

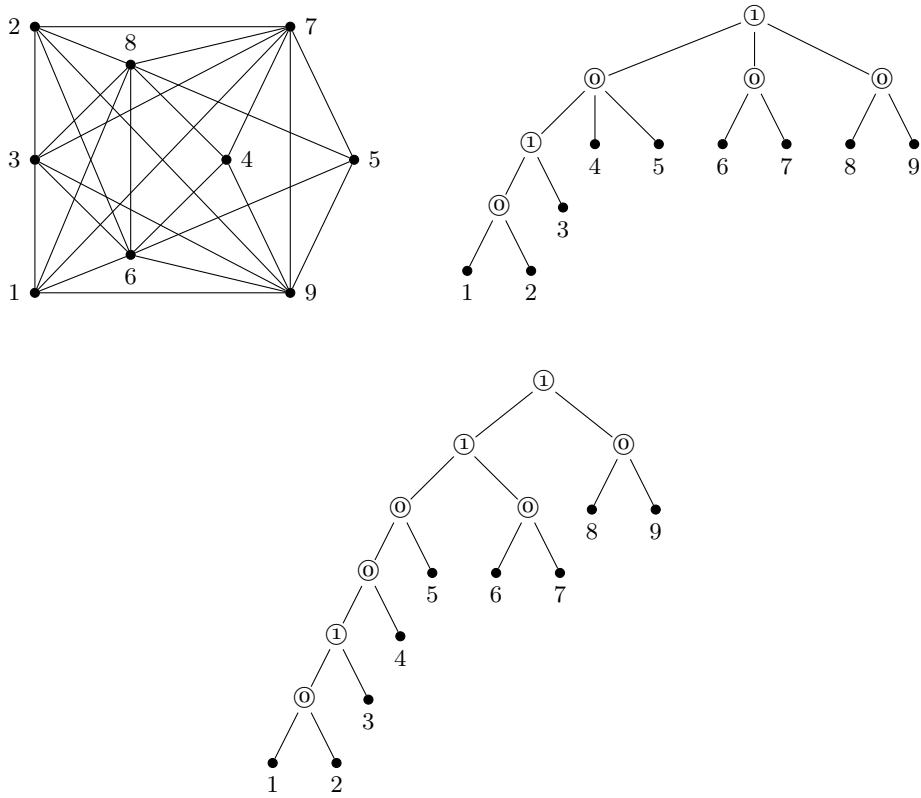
\begin{figure}[!ht]
\begin{center}
\tikzstyle{vertex}=[draw,circle,inner sep=1.2pt,fill=black] 

\begin{tikzpicture}[scale=.42] 
\node[vertex] (1) at (1,.8)  [label=left:{\small $1$}] {};
\node[vertex] (2) at (1,9.2)  [label=left:{\small $2$}] {};
\node[vertex] (3) at (1,5)  [label=left:{\small $3$}] {};
\node[vertex] (4) at (7,5)  [label=right:{\small $4$}] {};
\node[vertex] (5) at (11,5)  [label=right:{\small $5$}] {};
\node[vertex] (6) at (4,2)  [label=below:{\small $6$}] {};
\node[vertex] (7) at (9,9.2)  [label=right:{\small $7$}] {};
\node[vertex] (8) at (4,8) [label=above:{\small $8$}] {};
\node[vertex] (9) at (9,.8)  [label=right:{\small $9$}] {};
\draw (1)--(3)--(2);
\draw (6)--(8)--(7)--(9)--(6);
\draw (6)--(1)--(8); \draw (7)--(1)--(9);
\draw (6)--(2)--(8); \draw (7)--(2)--(9);
\draw (6)--(3)--(8); \draw (7)--(3)--(9);
\draw (6)--(4)--(8); \draw (7)--(4)--(9);
\draw (6)--(5)--(8); \draw (7)--(5)--(9);
\end{tikzpicture} 
\qquad
\begin{tikzpicture}[scale=.42] 
\node[vertex] (1) at (1,1)  [label=below:{\small $1$}] {};
\node[vertex] (2) at (3,1)  [label=below:{\small $2$}] {};
\node[vertex] (3) at (4,3)  [label=below:{\small $3$}] {};
\node[vertex] (4) at (5,5)  [label=below:{\small $4$}] {};
\node[vertex] (5) at (7,5)  [label=below:{\small $5$}] {};
\node[vertex] (6) at (9,5)  [label=below:{\small $6$}] {};
\node[vertex] (7) at (11,5)  [label=below:{\small $7$}] {};
\node[vertex] (8) at (13,5) [label=below:{\small $8$}] {};
\node[vertex] (9) at (15,5)  [label=below:{\small $9$}] {};
\node[knoten] (n0) at (2,3)  {\cojoin};
\node[knoten] (n1) at (3,5)  {\join};
\draw (1)--(n0)--(2); \draw (n0)--(n1)--(3);
\node[knoten] (n01) at (5,7)  {\cojoin};
\draw (n1)--(n01)--(4); \draw (n01)--(5);
\node[knoten] (n02) at (10,7)  {\cojoin};
\draw (6)--(n02)--(7);
\node[knoten] (n03) at (14,7)  {\cojoin};
\draw (8)--(n03)--(9);
\node[knoten] (n11) at (10,9)  {\join};
\draw (n01)--(n11)--(n02); \draw (n03)--(n11);
\end{tikzpicture} 

\vspace*{5ex}
\begin{tikzpicture}[scale=.42] 
\node[vertex] (1) at (1,1)  [label=below:{\small $1$}] {};
\node[vertex] (2) at (3,1)  [label=below:{\small $2$}] {};
\node[vertex] (3) at (4,3)  [label=below:{\small $3$}] {};
\node[vertex] (4) at (5,5)  [label=below:{\small $4$}] {};
\node[vertex] (5) at (6,7)  [label=below:{\small $5$}] {};
\node[vertex] (6) at (8,7)  [label=below:{\small $6$}] {};
\node[vertex] (7) at (10,7)  [label=below:{\small $7$}] {};
\node[vertex] (8) at (11,9) [label=below:{\small $8$}] {};
\node[vertex] (9) at (13,9)  [label=below:{\small $9$}] {};
\node[knoten] (n0) at (2,3)  {\cojoin}; 
\node[knoten] (n1) at (3,5)  {\join};
\draw (1)--(n0)--(2); \draw (n0)--(n1)--(3);
\node[knoten] (n01) at (4,7)  {\cojoin};
\draw (n1)--(n01)--(4); 
\node[knoten] (n02) at (5,9)  {\cojoin};
\draw (n01)--(n02)--(5);
\node[knoten] (n03) at (9,9)  {\cojoin};
\draw (6)--(n03)--(7);
\node[knoten] (n11) at (7,11)  {\join};
\draw (n02)--(n11)--(n03); 
\node[knoten] (n04) at (12,11) {\cojoin};
\draw (8)--(n04)--(9);
\node[knoten] (n111) at (9.5,13) {\join};
\draw (n11)--(n111)--(n04);
\end{tikzpicture} 
\caption{A cograph $G$, the cotree of $G$ and its binary version.}\label{fig:cotree}
\end{center}
\end{figure}

Now, given a cograph $G$ together with its binary cotree $T$, the bottom-up Algorithm~\ref{alg:sigma} below computes the cluster vertex deletion number $\varsigma(G)$ of $G$, as suggested by Lemma~\ref{lem:join-cojoin}. 
The algorithm traverses the cotree $T$ by post-order, that is, for the current node~$v$ of $T$, it recursively traverses the left subtree of $T_v$, then the right subtree of $T_v$, and finally visits the current node~$v$. The algorithm uses the following notations. For a node $v$ of $T$, 

\begin{itemize}
\item if $v$ is an internal node then $\ell(v)$ and $r(v)$ stands for the left child and the right child of~$v$, respectively;
\item $n(v)$ denotes the size of the subgraph of $G$ induced by the leaves of $T_v$. Thus, if $v$ is a leaf then $n(v)=1$ and if $v$ is the root of $T$ then $n(v)=|V(G)|$;
\item $\varsigma(v)$ denotes the cluster vertex deletion number of the subgraph of $G$ induced by the leaves of $T_v$. Thus, if $v$ is a leaf then $\varsigma(v)=0$ and if $v$ is the root of $T$ then $\varsigma(v)=\varsigma(G)$;
\item $\overline{\tau}(v)$ denotes the vertex cover number of the \emph{complement} of the subgraph of $G$ induced by the leaves of $T_v$. Thus, if $v$ is a leaf then $\overline{\tau}(v)=0$ and if $v$ is the root of $T$ then $\overline{\tau}(v)=\tau(\overline{G})$.
\end{itemize}

\begin{algorithm}[!ht]
\DontPrintSemicolon 
\KwIn{A cograph $G=(V,E)$ together with its (binary) cotree $T$.}
\KwOut{$\varsigma(G)$, the cluster vertex deletion number of $G$}

\medskip
Traverse $T$ by post-order and let $v$ be the current node\;
\If{$v$ is a leaf}{
    $n(v)\leftarrow 1$; 
    $\overline{\tau}(v)\leftarrow 0$;
    $\varsigma(v)\leftarrow 0$\;
    }
\Else{
       $n(v)\leftarrow n(\ell(v))+n(r(v))$\;
        \If{$v$ has label $\cojoin$}{
         $\overline{\tau}(v)\leftarrow \min\{\overline{\tau}(\ell(v))+n(r(v)), \overline{\tau}(r(v))+n(\ell(v))\}$\;
         $\varsigma(v)\leftarrow \varsigma(\ell(v))+\varsigma(r(v))$\;
         }
       \If{$v$ has label $\join$}{
         $\overline{\tau}(v)\leftarrow \overline{\tau}(\ell(v))+\overline{\tau}(r(v))$\;
         $\varsigma(v)\leftarrow \min\{\varsigma(\ell(v))+n(r(v)), \varsigma(r(v))+n(\ell(v)), \overline{\tau}(v)\}$\;
         }
}
\caption{\texttt{computing cluster vertex deletion number}}
\label{alg:sigma}
\end{algorithm}

\begin{proposition}\label{pro:P4free}
Given a  $P_4$-free $n$-vertex graph $G$ together with its cotree, Algorithm~\ref{alg:sigma} correctly computes the cluster deletion number $\varsigma(G)$ of $G$ in $O(n)$ time.
\end{proposition}
\begin{proof}
The correctness of Algorithm~\ref{alg:sigma} directly follows from Lemma~\ref{lem:join-cojoin}. 
Since per node in the cotree a constant number of operations is performed, the algorithm runs in $O(n)$ time.
\end{proof}

We remark that Algorithm~\ref{alg:sigma} can be slightly modified for computing a minimum cluster vertex deletion set. Also, since Lemma~\ref{lem:join-cojoin} holds accordingly for the weighted version, the minimum weight cluster vertex deletion number of cographs can be computed in linear time, too. 

\section{Computing the connected cluster vertex deletion number of cographs using the cotrees}\label{app:P4free-connected}
Recall that $S\subseteq V(G)$ is a connected cluster vertex deletion set if $G-S$ is a cluster graph and $G[S]$ is connected. 
Note that~$G$ has a connected cluster vertex deletion set if and only if~$G$ has at most one connected component that contains an induced~$P_3$ (if $G$ has more than two connected components containing an induced~$P_3$ then any cluster vertex deletion set must contain vertices in different connected components). 
Let $\varsigma_c(G)$ denote the connected cluster vertex deletion number of $G$, 
\begin{align*}
\varsigma_c(G)&=\min\{|S| :  \text{$S$ is a connected cluster vertex deletion set of $G$}\}. 
\end{align*}
(We set $\varsigma_c(G)=\infty$ if $G$ has no connected cluster vertex deletion set.) 

When computing $\varsigma_c(G)$, we will have to consider a special case of (connected) cluster vertex deletion.  A set $S\subseteq V(G)$ is a \emph{(connected) clique deletion set} if $G-S$ is a clique (and $G[S]$ is connected). Let $\theta(G)$ and $\theta_c(G)$ denote the clique vertex deletion number and the connected clique vertex deletion number of $G$, respectively,  
\begin{align*}
\theta(G)&=\min\{|S| : \text{$S$ is a clique deletion set of $G$}\},\\ 
\theta_c(G)&=\min\{|S| :  \text{$S$ is a connected clique deletion set of $G$}\}. 
\end{align*}
(Again, we set $\theta_c(G)=\infty$ if $G$ has no connected clique deletion set.) 
Notice that $\theta(G)=\tau(\overline{G})$, and thus $\theta(G)$ can be computed in linear time when restricted to cographs (by Lemma~\ref{lem:join-cojoin} and Proposition~\ref{pro:P4free}.) 
Notice also that $\theta(G)\le\theta_c(G)$ and $\varsigma(G)\le\varsigma_c(G)$.  
We will see in this section that $\theta_c(G)$ and $\varsigma_c(G)$ can be computed efficiently when restricted to cographs. 

We first consider the connected clique vertex deletion number. The following fact follows immediately from the definition:
\begin{lemma}\label{lem:conCliqueDcojoin}
For arbitrary graphs $G_1$ and $G_2$,
\[
\theta_c(G_1\cojoin G_2)=\begin{cases}
   \infty,&\text{if $G_1$ or $G_2$ is disconnected, or both $G_1, G_2$}\\
           &\text{are non-complete;}\\
   |V(G_1)|,&\text{if $G_2$ is a complete and $G_1$ a connected}\\
                &\text{non-complete graph;}\\
   |V(G_2)|,&\text{if $G_1$ is a complete and $G_2$ a connected}\\ 
                &\text{non-complete graph;}\\
   \min\{|V(G_1)|,|V(G_2)|\},&\text{if $G_1$ and $G_2$ are complete graphs.}
\end{cases}
\]
\end{lemma}

The following two lemmas provide a formula for computing the connected clique vertex deletion number of the join of two graphs.

\begin{lemma}\label{lem:conCliqueDjoin-a} 
Let $G_1$ be a complete graph and let $G_2$ be an arbitrary graph. 
Then:
\[
\theta_c(G_1\join G_2) = \min\left\{\theta_c(G_2), 1+\theta(G_2)\right\}.
\]
\end{lemma}
\begin{proof}
Let  $S$ be  an optimal connected clique vertex deletion set of $G_1\join G_2$, and write $S_i=S\cap V(G_i)$, $i=1,2$. Then $S_1$ is a (connected) clique deletion set of $G_1$ (possibly empty) and~$S_2$ is a clique deletion set of $G_2$. 
Thus, $|S_2|\ge\theta(G_2)$. Moreover, if $G_2[S_2]$ is connected then $|S_2|\ge\theta_c(G_2)$, and hence in this case, $\theta_c(G_1\join G_2)=|S|=|S_1|+|S_2|\ge \theta_c(G_2)$. 
If $G_2[S_2]$ is disconnected then $|S_1\cap V(G_1)|=1$ (due to the connectedness and the optimality of~$S$) and $|S|\ge 1+\theta(G_2)$. 
Hence, in this case, $\theta_c(G_1\join G_2)=|S|\ge 1+\theta(G_2)$.

For the other direction, let $S$ be a clique vertex deletion set of $G_2$ of optimal size~$\theta(G_2)$. If $G_2[S]$ is connected then $S$ is a connected clique deletion set of $G_1\join G_2$, hence $\theta_c(G_1\join G_2)\le |S|=\theta_c(G_2)$.  
If $G_2[S]$ is disconnected then, for any vertex $u\in V(G_1)$, $S\cup\{u\}$ is a connected clique deletion set of $G_1\join G_2$, hence $\theta_c(G_1\join G_2)\le |S\cup\{u\}| =1+\theta(G_2)$.
\end{proof}

\begin{lemma}\label{lem:conCliqueDjoin-b} 
Let $G_1$ and $G_2$ be two arbitrary non-complete graphs. 
Then:
\[
\theta_c(G_1\join G_2) = \theta(G_1)+\theta(G_2).
\]
\end{lemma}
\begin{proof}
Let $S$ be an optimal connected clique deletion set of $G_1\join G_2$ and write $S_i=S\cap V(G_i)$, $i=1,2$. Then $S_i$ is a clique deletion set of $G_i$, hence $|S|=\theta_c(G_1\join G_2) = |S_1|+|S_2|\ge \theta(G_1)+\theta(G_2)$. 

For the other direction let $T_i$ be an optimal clique deletion set of $G_i$, $i=1,2$. By assumption, $T_i\not=\emptyset$, hence $T_1\cup T_2$ is a connected clique deletion set of $G_1\join G_2$. Therefore, $\theta_c(G_1\join G_2) \le |T_1|+|T_2|= \theta(G_1)+\theta(G_2)$. 
\end{proof}

We now consider the connected cluster vertex deletion number of the disjoint union and the join of two graphs. 
The following fact follows immediately from the definition:
\begin{lemma}\label{lem:conCVDcojoin}
For arbitrary graphs $G_1$ and $G_2$,
\[
\varsigma_c(G_1\cojoin G_2)=\begin{cases}
   \infty,&\text{if $G_1$ or $G_2$ has two non-clique components, or both $G_1,G_2$}\\
           &\text{are not $P_3$-free;}\\
   \varsigma_c(C),&\text{if one of $G_1$ and $G_2$ is $P_3$-free and $C$ is the unique non-clique}\\
        &\text{component of the other;}\\
   0,&\text{if $G_1$ and $G_2$ are $P_3$-free.}
\end{cases}
\]
\end{lemma}

Lemmas~\ref{lem:conCVDjoin-a} and~\ref{lem:conCVDjoin-c} below provide a formula for computing the connected cluster vertex deletion number of the join of two graphs.

\begin{lemma}\label{lem:conCVDjoin-a} 
Let $G_1$ be a complete graph and let $G_2$ be an arbitrary graph. 
Then:
\[
\varsigma_c(G_1\join G_2) = \min\left\{|V(G_1)|+\varsigma(G_2), \theta_c(G_2), 1+\theta(G_2)\right\}.
\]
\end{lemma}
\begin{proof} 
Let  $S$ be  a connected cluster vertex deletion set of $G_1\join G_2$ of optimal size, and write $S_i=S\cap V(G_i)$, $i=1,2$. Then $S_1$ is a (connected) clique deletion set of $G_1$ (possibly empty) and $S_2$ is a cluster vertex deletion set of $G_2$. Moreover, if $G_2-S_2$ is not a clique then $S_1=V(G_1)$, hence $|S|=\varsigma_c(G_1\join G_2)\ge |V(G_1)|+\varsigma(G_2)$.
In the case where $G_2-S_2$ is a clique, $|S_2|\ge\theta(G_2)$. Moreover, if $G_2[S_2]$ is connected then $S_1=\emptyset$ (because of the optimality of $S$) and $|S_2|\ge\theta_c(G_2)$; if $G_2[S_2]$ is disconnected, $|S_1\cap V(G_1)|=1$. 
Hence in this case, $|S|=\varsigma_c(G_1\join G_2)=|S_1|+|S_2|\ge\min\left\{\theta_c(G_2), 1+\theta(G_2)\right\}$.

For the other direction, observe first that by definition, $\varsigma_c(G_1\join G_2)\le \theta_c(G_1\join G_2)$, and hence by Lemma~\ref{lem:conCliqueDjoin-a}, $\varsigma_c(G_1\join G_2)\le\min\left\{\theta_c(G_2), 1+\theta(G_2)\right\}$. 
Observe next that, for any cluster vertex deletion set $S$ of $G_2$ of optimal size $\varsigma(G_2)$, $V(G_1)\cup S$ is a connected cluster vertex deletion set of $G_1\join G_2$, hence $\varsigma_c(G_1\join G_2)\le |V(G_1)|+\varsigma(G_2)$.
\end{proof}

For two non-complete graphs, we first show:
\begin{lemma}\label{lem:conCVDjoin-b} 
Let $G_1$ and $G_2$ be two arbitrary, non-complete graphs. Then:
\[
\varsigma_c(G_1\join G_2) \ge \min\left\{|V(G_1)|+\varsigma(G_2), |V(G_2)|+\varsigma(G_1), \theta(G_1)+\theta(G_2)\right\}.
\]
Furthermore, if both $G_1$ and $G_2$ are disconnected, then:
\[
\varsigma_c(G_1\join G_2) \ge \min\left\{|V(G_1)|+\max\{\varsigma(G_2),1\}, |V(G_2)|+\max\{\varsigma(G_1),1\},\theta(G_1)+\theta(G_2)\right\}.
\]
\end{lemma}
\begin{proof}
Let $S$ be  a connected cluster vertex deletion set of $G_1\join G_2$ of optimal size, and write $S_i=S\cap V(G_i)$, $i=1,2$. Then $S_i$ is a cluster vertex deletion set of $G_i$. Note, moreover, that at least one of $G_1-S_1$ and $G_2-S_2$ must be a clique.

If each of $G_1-S_1$ and $G_2-S_2$ is a clique, $S_1$ and $S_2$ are clique deletion sets of~$G_1$ and~$G_2$, respectively.
Hence in this case, $|S|=\varsigma_c(G_1\join G_2)=|S_1|+|S_2|\ge\theta(G_1)+\theta(G_2)$. 
If $G_1-S_1$ is not a clique then $S_2=V(G_2)$, and likewise, if $G_2-S_2$ is not a clique then $S_1=V(G_1)$. 
In these two cases, $|S|=\varsigma_c(G_1\join G_2)\ge \min\{|V(G_1)|+\varsigma(G_2), |V(G_2)|+\varsigma(G_1)\}$.

Now, suppose that both $G_1$ and $G_2$ are disconnected. Then, the connectivity of $S$ implies that if $S_1=V(G_1)$ then $|S_2\cap V(G_2)|\ge 1$, and likewise, if $S_2=V(G_2)$ then $|S_1\cap V(G_1)|\ge 1$. Hence, $|S|=\varsigma_c(G_1\join G_2)\ge \min\{|V(G_1)|+\max\{\varsigma(G_2),1\}, |V(G_2)|+\allowbreak\max\{\varsigma(G_1),\allowbreak 1\}\}$.
\end{proof}

\begin{lemma}\label{lem:conCVDjoin-c} 
Let $G_1$ and $G_2$ be two arbitrary, non-complete graphs. 
\begin{itemize}
\item[\em (1)] If $G_1$ or $G_2$ is connected, then:
\[
\varsigma_c(G_1\join G_2) = \min\left\{|V(G_1)|+\varsigma(G_2), |V(G_2)|+\varsigma(G_1), \theta(G_1)+\theta(G_2)\right\}.
\]
\item[\em (2)] If both $G_1$ and $G_2$ are disconnected, then:
\[
\varsigma_c(G_1\join G_2) = \min\left\{|V(G_1)|+\max\{\varsigma(G_2),1\}, |V(G_2)|+\max\{\varsigma(G_1),1\},\theta(G_1)+\theta(G_2)\right\}.
\]
\end{itemize}
\end{lemma}
\begin{proof} 
By Lemma~\ref{lem:conCVDjoin-b}, it remains to show that in both claims the left-hand side is at most the right-hand side. Observe first that $\varsigma_c(G_1\join G_2)\le\theta_c(G_1\join G_2)$, and so by Lemma~\ref{lem:conCliqueDjoin-b}, $\varsigma_c(G_1\join G_2)\le\theta(G_1)+\theta(G_2)$. 

\medskip\noindent
(1): Let $G_1$ be connected, say. 
Observe that any cluster vertex deletion set $S_1$ of $G_1$ is non-empty (because~$G_1$ is connected non-complete), hence $V(G_2)\cup S_1$ is a connected cluster vertex deletion set of $G_1\join G_2$, and for any cluster vertex deletion set $S_2$ of $G_2$, $V(G_1)\cup S_2$ is a connected cluster vertex deletion set of $G_1\join G_2$ (because~$G_1$ is connected). 
Thus, 
$\varsigma_c(G_1\join G_2) \le\min\{|V(G_1)|+\varsigma(G_2), |V(G_2)|+\varsigma(G_1)\}$. 

\medskip\noindent
(2): 
Observe that for any cluster vertex deletion set $S_1$ of $G_1$ of optimal size $\varsigma(G_1)$, $V(G_2)\cup S_1$ (if $S_1\not=\emptyset$) or $V(G_2)\cup\{u\}$ (if $S_1=\emptyset$), where $u$ is any vertex of $G_1$,  is a connected cluster vertex deletion of $G_1\join G_2$. Hence $\varsigma_c(G_1\join G_2)\le |V(G_2)|+\max\{\varsigma(G_1),1\}$. Similarly, $\varsigma_c(G_1\join G_2)\le |V(G_1)|+\max\{\varsigma(G_2),1\}$.
\end{proof}

Now, given a cograph $G$ together with its cotree, with Lemmas~\ref{lem:conCliqueDcojoin}, \ref{lem:conCliqueDjoin-a}, \ref{lem:conCliqueDjoin-b}, \ref{lem:conCVDcojoin}, \ref{lem:conCVDjoin-a} and~\ref{lem:conCVDjoin-c} we can compute the connected clique vertex deletion number and the connected cluster deletion number of $G$ in linear time. This is done in the same way for computing the vertex cover number and the cluster vertex deletion number in Appendix~\ref{app:P4free}, hence we omit the details.

\end{document}